\documentclass[twocolumn,10pt,twoside]{IEEEtran}
\usepackage{amsmath}
\usepackage{amsthm,amssymb,amsmath,bm}
\usepackage{subfigure}
\usepackage{amsfonts}
\usepackage{amsmath}

\usepackage{epsfig}
\usepackage{amssymb}
\usepackage{amsmath}
\usepackage{cite}
\usepackage[utf8x]{inputenc}
\usepackage{multirow}
\usepackage{rotating}
\usepackage{algorithmic}
\usepackage{graphicx}
\usepackage{tabularx}
\usepackage{array}
\usepackage{color,soul}
\usepackage{bm}
\usepackage{tikz}
\usepackage{graphicx,dblfloatfix}
\usepackage{blindtext}

\newtheorem{remark}{Remark}
\usepackage{subfigure}
\usepackage{moreverb}
\usepackage{epsfig}
\usepackage{amsmath,amssymb,amsthm,mathrsfs,amsfonts,dsfont}
\usepackage{adjustbox,lipsum}
\usepackage{amsfonts}
\usepackage{epsfig}
\usepackage{amssymb}
\usepackage{amsmath}
\usepackage{amsthm}
\usepackage{subfigure}
\usepackage{multirow}
\usepackage{rotating}
\usepackage{graphicx}
\usepackage{tabularx}
\usepackage{array}
\usepackage{anyfontsize}
\usepackage{color,soul}
\usepackage{graphicx,dblfloatfix}
\usepackage{epstopdf}
\usepackage{blindtext}
\usepackage{amsmath}
\usepackage{amsthm,amssymb,amsmath,bm}
\usepackage{subfigure}
\usepackage{amsfonts}
\usepackage{epsfig}
\usepackage{amssymb}
\usepackage{amsfonts}
\usepackage{amsmath}
\usepackage{cite}
\usepackage{graphicx}
\usepackage{fancyhdr}
\usepackage{subfigure}
\usepackage[subfigure]{tocloft}
\usepackage[font={small}]{caption}
\usepackage{subfigure}
\usepackage{tabularx}
\usepackage{tcolorbox}
\usepackage{cite}
\usepackage[linesnumbered,ruled,vlined]{algorithm2e}
\SetKwInput{KwInput}{Input}
\SetKwInput{KwOutput}{Output}
\usepackage{amsthm,amssymb,amsmath,bm}
\usepackage{graphicx}
\usepackage{fancyhdr}
\usepackage{subfigure}
\usepackage[subfigure]{tocloft}
\usepackage[font={small}]{caption}
\usepackage{subfigure}
\usepackage{tabularx}
\usepackage{cite}
\allowdisplaybreaks
\usepackage{url}
\usepackage{hyperref}
\newtheorem{theorem}{Definition}
\newtheorem{proposition}{Proposition}

\begin{document}
\vspace{-10mm}
\title{ \huge Robust Energy-Efficient Resource Management, SIC Ordering, and Beamforming Design for MC MISO-NOMA Enabled 6G 
	%Robust Energy-Efficient for PD-NOMA Enabled 6G Networks via Global SIC Ordering and Beamforming Design  
}
\vspace{-10mm}
	% High Degrees of Reused Freedom  Beamforming Design  in 
%	MISO-NOMA-assisted C-RAN Network: }
%		\title{ \huge  Robust Energy-Efficiency SIC Ordering and  Beamforming Design  in 
%			MISO-NOMA-assisted C-RAN Network: High Degrees of Reused Freedom}
		% Under Imperfect CSI}
		\author{Abolfazl Zakeri, \textit{Student Member, IEEE},~Ata Khalili,~\textit{ Member, IEEE},~Mohammad Reza Javan,~\textit{Senior Member, IEEE},~Nader Mokari,~\textit{Senior Member, IEEE}, and  Eduard A Jorswieck, \textit{Fellow, IEEE}
		\thanks{	
		A. Zakeri, 
		%is with the Centre for Wireless Communications-Radio Technologies, University of Oulu, 90014 Oulu, also he was in Trabiat Modares University
			%	(email: abolfazlzakeri@gmail.com).
		 A. Khalili and N. Mokari are with the Department of Electrical and Computer Engineering,~Tarbiat Modares University,~Tehran,~Iran, (e-mails: ata.khalili@ieee.org, abolfazlzakeri@gmail.com, and nader.mokari@modares.ac.ir). M.~R.~Javan is with the Department of Electrical Engineering, Shahrood University of Technology, Iran, (e-mail: javan@shahroodut.ac.ir).~Eduard A. Jorswieck is with Institute for Communications Technology, TU
			Braunschweig, Germany, Email: jorswieck@ifn.ing.tu-bs.de.
		}}

	% <-this % stops a space
%	\markboth{	IEEE Trans. Veh. Technol}%}%
	%\markboth{IEEE Trans. Commun}%
	%{Submitted paper}
	\maketitle
	\vspace{-7mm}
	\begin{abstract}
%Successive interference cancellation (SIC) in non-orthogonal multiple access (NOMA) has become the bottleneck which limits the performance improvement for transmission data.~
This paper studies a novel approach for successive interference cancellation (SIC) ordering and beamforming in a  multiple antennas  non-orthogonal multiple access (NOMA) network with multi-carrier multi-user setup. To this end, we formulate a joint beamforming design, subcarrier allocation, user association,
and SIC ordering algorithm to maximize the worst-case energy efficiency (EE). The formulated problem is a non-convex mixed integer non-linear programming (MINLP)  which is generally
difficult to solve. To handle it, we first 
adopt the  linearizion technique as well as relaxing the
integer variables, and then we
employ the Dinkelbach algorithm to convert it into a more mathematically tractable form. The adopted non-convex optimization problem is transformed into
an equivalent rank-constrained semidefinite programming (SDP) and
is solved by SDP relaxation and exploiting  sequential fractional programming. 
%
%
%
%~First, we derive the worst-case data rate and SIC ordering constraint based on the norm-bound matrix.~Second,~we transform the fractional problem into a more mathematically tractable one based on the Dinkelbach algorithm.~The resulting problem is non-convex mixed integer non-linear programming (MINLP) which is generally difficult to solve.~To tackle this issue,~we first employ the linearizion technique as well as relaxing the integer variables, and then, we employ the majorization minimization approach to handle the difficulty of the problem.~The adopted non-convex optimization problem is transformed into
%an equivalent rank-constrained semidefinite program (SDP) and
%is solved by SDP relaxation and exploiting  sequential fractional programming. 
Furthermore, to strike a balance between complexity and performance,
a low complex  approach based on 
 alternative optimization is adopted. 
Numerical results unveil that the proposed SIC ordering method outperforms the conventional existing works addressed in the literature.  
%Besides results demonstrate that the proposed robust EE algorithm can alleviate negative effect of imperfect of CSI, SIC, and limited power and spectrum resources on the network performance.  
	\end{abstract}
%\newpage
%\vspace{-1.5em}
	\section{Introduction}
	\subsection{Motivations and State of the Art}
In the recent decades, wireless communications have appealed to a growing number of customers, demanding high quality services, and ubiquitous connections.~In order to fulfill these demands, 
%recent development of 5G networks cannot satisfactory 
the next generation of wireless networks, namely sixth-generation (6G)\footnote{Recently,  fifth generation of wireless networks is deployed and its evolution towards 6G has been started  \cite{6g1}.},
%Recently rollouted  network cannot efficiently support all ever new demands \cite{6g1}.},
should be redesigned and exploit advanced technologies \cite{6G}. 
%support massive connection and high data rate demanded services \cite{6G}.
%\textcolor{blue}
Regarding this, the network must be designed in such a way to dynamically change its architecture and the communications technologies.~In such
a flexible architecture, significant amount of signaling and computational
resources are needed to optimally manage
the network resources and enable to design the flexible resource sharing. Recently, various radio access network (RAN) architectures as distributed and centralized RAN (C-RAN) are developed to provide an efficient computational resource sharing and  resource utilization \cite{RAN}.
%		\textcolor{red}{The first one is a new air interface architecture by means of separating signaling and data to have
%		an efficient and flexible radio resource management (RRM)
%		for capacity boosting and energy saving. The second one is
%		RAN mode selection with renovating it into massive base stations (BSs) with
%		a centralized baseband processing. 
%		%This perspective motivates us to embed the cloud-based baseband processing pool with remote radio heads (RRHs) and BS functions virtualization.
%		The third one is separating the control plane from the data plane in order to have an efficient centralized RRM with a global
%		view of the network.~In addition to the RAN architecture, the transmission technologies
%		used in RAN have important effect on the efficient
%		resource management. Various technologies such as non-orthogonal
%		multiple access (NOMA) and
%		(massive) multiple input multiple output (MIMO) are proposed
%		to improve the performance gain \cite{NOMAMIMO}.
%	}
	%	\textcolor{blue}{ 
			To enable sustainable 6G networks,~new emerging techniques  such as new multiple access (MA)   and multiple antennas systems (MAS) are needed to  improve the network performance (e.g., energy efficiency (EE))
			 \cite{Ding_1}.~In this regard, non-orthogonal
			multiple access (NOMA) and multiple-input single-output (MISO)  are promising approaches which can significantly improve EE and provide massive connectivity applications, i.e., Internet of Thing (IoT)  compared to the orthogonal multiple access (OMA) and single antenna systems\cite {Ding_2,Ding_3,Yang,Dai}.~In fact,~improving the EE, fairness, and flexibility in resource allocation have turned to apply NOMA  systems as the main trend in the beyond current wireless network.
			%		 		~In NOMA, each resource block can be allocated to multiple users in the coverage area of each BS.  
	The authors in \cite{Dai} present a basic principle of NOMA in which they state a systematic comparison among the different NOMA techniques from the viewpoint of the EE and receiver complexity.~In particular, NOMA utilizes power domain (PD-NOMA) based networks for MA as well as successive interference cancellation (SIC) which removes the undesired multiuser interference\cite{Islam}.~In NOMA, SIC ordering is one of the key challenges and it is critical for the performance of data transmission to handle the NOMA interference \cite{Unveal_SIC1,Unveal_SIC2}. 
			However, the SIC ordering problem has not been addressed well, 
			and there are open problems that need to be addressed properly \cite{Unveal_SIC1,Unveal_SIC2}. In fact, in most of the works, 
			SIC ordering is considered based on the channel gain which is not practical and optimal due to necessity of full channel state information (CSI) \cite{Ro2,Dai,FNOMA,MISO_Optimal,Moltafet,Minorization,ICC,Sharma}. 
%			 However, in practice providing full CSI 
%			  which is
%			 not accurate and practical representation of the real wireless communication networks.
%			 \cite{Ro2,Dai,FNOMA,KNOMA,MISO_Optimal,Moltafet,Minorization,ICC,Sharma}. 
			%
			Besides, there is uncertainty in the CSI that cannot be applicable for multiple antennas systems. 
			%Besides, there are a few studies on optimizing beamforming and user scheduling (association and subcarrier assignment) in HetNets, simultaneously. Also, these are no global\footnote{The term of global refer to centralized interference management in NOMA-based HetNets.} framework on designing SIC for both intra-cell and inter-cell treatment.
%	}
	%	~In particular, SIC is an effective method to harness the interference.
		\\\indent
	This paper proposes a worst-case SIC ordering, resource allocation policy, and beamforming design for multicarrier (MC) MISO-NOMA networks.
		%  while considering  SIC ordering based on the network conditions such as available   resource.~
		We would like to see how much  we can get performance gain in MAS for the SIC ordering as compared to OMA as well as traditional methods in which SIC ordering is based on the channel gains.
		%~How the performance gain of the system is compared with traditional SIC order based on the channel gains?    
%		\vspace{-1.4em}
	\subsection{Related Works}
		In NOMA, the efficient allocation of scarce resources and SIC ordering are turned into a challenging necessity
		for improving users' satisfaction \cite{Unveal_SIC1}.~There
	are some attempts to find proper resource allocation
strategies which improves the overall performance of such networks\cite{Wei,Sharma,FNOMA,Moltafet,Minorization,ICC}.~For instance, the problem of power allocation and precoding design is proposed in \cite{Minorization} in which they employ single carrier multiple-input multiple-output (MIMO) NOMA systems.~The authors in \cite{FNOMA} propose an optimal resource allocation to maximize the system throughput for NOMA and full-duplex (FD) systems,~respectively. ~However,~the base station (BS) is equipped with a single antenna which cannot fully exploit the degrees of freedom of the network.~The
		works in \cite{MISO_Optimal,Fairness,Robust,Zakeri2} consider  beamforming design for MISO-NOMA systems to optimize the performance and cost of the system. In particular, the authors in \cite{Robust} propose a robust
		beamforming design for MISO-NOMA system to maximize
		the minimum data rate. In \cite{MISO_Optimal}, beamforming design and subcarrier allocation for maximizing the total data
		rate are proposed where optimal and sub-optimal solutions
are provided.~The beamforming design for maximizing the
		minimum EE and proportional fairness are
		developed in \cite{Fairness} to strike a balance between the EE of the system and the fairness between users.~Most of the previous works considered the fixed SIC ordering,
		in which the order of decoding at each receiver is determined
according to the channel gains \cite{FNOMA,MISO_Optimal,Moltafet,Minorization,ICC,Sharma}.~In SIC, the users are ordered and each user can remove the interference from users determined by the ordering scheme.~Although most of the works on the SIC ordering sort the users based on their channels, this is not a practical scenario and can not be guaranteed as well.~Also, it should be noted that sorting users for SIC based on the channel gains is neither 
 optimal nor practical scheme at all, especially in MAS due to unavailability of the full CSI
\cite{MISO_Optimal}.~To circumvent this problem, SIC ordering should be based on the network, channel gain, and the available resource conditions.
The authors in \cite{Zakeri} address this problem for single antenna BS and perfect CSI scenario which is not practical and appropriate for future networks due to considering single-antenna BS and also having perfect CSI channel. 

\textcolor{blue}{Besides, EE is an important metric for wireless networks, especially for enabling green communication. 
New communication technologies are proposed to improve the system EE.} In particular, various techniques are proposed which aim to enhance the network throughput while consuming less energy without sacrificing the quality of service (QoS).~\textcolor{blue}{At the same time, EE maximization problems are indispensable in NOMA systems, to strike a good throughput-power tradeoff and  improve the system performance which is noticed as one of the key performance metrics in future wireless networks.}~However, there is a deficit of existing works on the literature considering the EE.~For instance, in \cite{Y. Zhange}, EE maximization is studied to obtain an optimal power allocation based on the non-linear
	fractional programming method. Furthermore, in \cite{Fang},
	a subchannel assignment and power allocation is investigated to maximize the EE in NOMA networks.~However, in real scenarios assuming perfect CSI is not a valid assumption due to some issues like quantization and channel estimation errors as well as hardware limitations.~In this regards the works in \cite{Robust,Ro2,Ro3,FangJ} address the robust solution for imperfect CSI. In particular, in \cite{Robust,Ro2,Ro3},  robust designs for the MISO-NOMA systems are developed based on the bounded channel uncertainties. The joint user scheduling and power allocation are explored in \cite{FangJ} while considering  imperfect CSI.~User association in multi-cell NOMA systems is also challenging. Specifically, in addition to the NOMA interference caused by the co-channel interference, the interference between cells 
	%causes through un-associated BSs 
	also needs to be taken into consideration
%which mitigates the interference between cells
 \cite{Comp,User}. 
Nonetheless to the best of the authors knowledge, the problem of beamforming design and SIC ordering in a MC MISO-NOMA enabled C-RAN network while considering imperfect CSI has not been investigated yet. In \cite{FNOMA,Zakeri,Minorization,Moltafet,ICC,Sharma}, the BS is equipped with single antenna while assuming perfect CSI.~The works in \cite{Ro2,Ro3,Robust} consider robust beamforming design while fixed SIC ordering.~Furthermore,~the authors in \cite{Comp,User} consider user association for the single antenna BS while SIC is based on the channel gains. In addition, in \cite{Zakeri}, the SIC ordering problem  for single antenna BS and perfect CSI is considered.~Consequently, user association policy and SIC ordering in a MISO-NOMA enabled C-RAN network  with imperfect CSI are still open problems which have not been addressed yet.
%\vspace{-1.5em}
	\subsection{ Contributions and Research Outcomes}
In this paper, we aim to bridge the above mentioned knowledge gap.~In particular,~we propose a joint beamforming design, subcarrier allocation, user association, and SIC ordering algorithm which maximizes the EE of the network under imperfect CSI. 
		To this end, we formulate the problem of beamforming design and SIC ordering to maximize the worst-case system EE.~In our method, SIC ordering is considered as an optimization variable while in more  existing works, SIC ordering is fixed and depends on the channel gains.
		%~In other words, each BS determines the user's decoding order by solving the proposed optimization problem.~
		The optimization problem is a non-convex mixed integer non-linear programming which is very difficult to solve.~To handle it, we employ majorization minimization (MM), abstract Lagrangian method, semi-definite relaxation (SDR) method, and sequential fractional programming to handle the beamforming design and integer variables.

%	%_______________________________________________________________
%	%However, such simultaneous usage of the same spectrum cAAUses severe interference.
%	 	
%\textcolor{blue}
%{	
Our main contributions are summarized as follows:
	\begin{itemize}
%		\item
%	We propose a novel model for the user association based on the multi-connectivity in which each user can be assigned to different active antenna units (AAUs) in different subcarriers to increase the degrees of freedom of the network.
	\item 	We propose a novel SIC ordering method for the downlink of a MC MISO NOMA. To this end, we formulate a novel optimization problem to maximize the EE by performing the subcarrier allocation, beamforming design, user association, and SIC ordering.~In particular, we formulate a new problem to investigate how to order users to apply
	successful SIC based on the available resources. Also, we derive the worst-case SIC ordering condition as an optimization constraint and then tackle its non-convexity. 
	%	\item
	%	We propose a new SIC ordering method for downlink of the PSMA-based system. To this end, we formulate a novel  optimization problem in which the main aim is to maximize the sum rate over subcarrier allocation, transmit power, and SIC ordering parameter.
		\item
	We study the  practical imperfect CSI in C-RAN networks.~In doing so, we consider the worst-case EE to provide a robust resource allocation algorithm. 
		\item
	We propose a solution based on rank-constrained semidefinite programming (SDP) relaxation and exploiting  sequential fractional programming. In particular, we adopt
	MM approach and penalty factor to make it mathematically tractable and then we adopt Dinkelbach algorithm. Moreover, we provide a low complexity iterative algorithm in which the scheduling variable, i.e., user association and subcarrier assignment, is obtained through the matching algorithm.
	\item	
Numerical	results  reveal that the proposed worst-case EE maximization and SIC ordering algorithm can alleviate negative effect of imperfect of CSI, SIC, and limited power and spectrum resources on the network performance. Also, the results showcase  the superiority of the proposed algorithm compared to the other conventional schemes.   
	\end{itemize}
\indent The rest of this paper is organized as follows. The system model and problem formulation is discussed in Sec. \ref{systemmodelandproblemformulation}. The solution algorithm and complexity analysis  are presented in Sec. \ref{solutions}. Finally, the simulation analysis and conclusions are provided in Secs. \ref{Simulation_results} and \ref{Conclusion}, respectively.
\\\indent
	\textbf{Notations}: Vector and matrix
variables are indicated by bold lower-case and upper-case letters,
respectively.  $|.|$ indicate the absolute value, $ \|.\|$ or $ \|.\|_2$ denotes the  Euclidean norm ($\textit{l}_2$ norm), and $\bold{A}^\dagger$ and $ \bold{A}^{T}$ indicate the conjugate transpose and transpose of matrix $ \bold{A}$, respectively. \textcolor{blue}{ Also, $\text{Tr}[\bold{A}]$ denotes the trace of matrix $ \bold{A}$ and $\bold{I}_{M}$ denotes the $M \times M$ identity matrix. $a^{*}$ denotes the optimal value of variable $a$. $\nabla_{\bold{x}}g$ indicate the gradient vector of function $g(\bold{x})$. }
$\mathcal{S}$ denotes the set $\{1,2,\dots,S\}$ and $|\mathcal{S}|=S$ is the cardinality of set $\mathcal{S}$. $\mathcal{S}\backslash\{s\}$ discards the element $s$ from the set $\mathcal{S}$. $\mathcal{C}^{M\times 1}$ denotes the set of $ M $-by-$ 1 $ dimensional complex vectors, and operation $  \Bbb{E}\{.\} $  denotes the statistical expectation. 
%\vspace{-1em}
	\section{\textcolor{blue}{ System Model and Problem Formulation}}\label{systemmodelandproblemformulation}
	%complex matrix.
	%\textcolor{black}{In this paper, we focus on the downlink of a PSMA-based network.
	%	Due to the fact that the SIC ordering algorithm in both the PSMA and PD-NOMA  techniques are same, for completeness, we also evaluate the performace of the PD-NOMA-based network.}
		%consider the downlink of PD-NOMA and PSMA-based networks. Due to better performance  We also study PD-NOMA-based network and  . }
	\subsection{{ System Model Descriptions and Related Constraints}}
	 		In this paper, we consider a downlink scenario for a C-RAN consisting of  a set  of $ F $ \textcolor{blue}{active antenna units (AAUs)} indexed by $ f $, whose set is denoted by $\mathcal{F}=\{1,\dots,{F}\}$, where $f=1$ is a high power AAU,  and a base band unit (BBU).
	 	Let	$P_{\text{max}}^{f}$ be the transmit power budget of AAU $f$.
 	 Each AAU is	equipped  with $M$ antennas, uses a set $\mathcal{N}$ of $N$ shared subcarriers,  is connected  to the BBU with a limited bandwidth fronthaul/metro-edge link, and
	 		% Each AAU $f$ shares considered $N$ subcarriers  with equal bandwidth 
	 		utilizes PD-NOMA  to transmit data to  
	 		 single antenna  end-users.  In fact, we consider a MC MISO-NOMA communication network setup.
 		  We denote the set of  all users as $\mathcal{K}=\{1,\dots,K\}$   which are randomly  distributed with the uniform
 		  distribution  inside the coverage/service area of the network \cite{User,Zakeri2}.
	 	 The  considered system model is depicted in Fig. \ref{PSMA}. 
	 	 %In such network we have $ \rho_{K,n} $
	 	 Considering that user $ 2 $ performs SIC on users $ 1 $ and $ 3 $ over subcarrier $ 2 $,  we have such output for SIC ordering variable which will be explained in Sec. \ref{SIC_Section}, $ \xi^2_{1,2}=1,  \xi^2_{3,2}=1,  \xi^3_{2,2}=0,$ and $  \xi^1_{2,2}=0$. The definition of main notations are listed in Table \ref{Notations}.  
	 	 %As an example, assume  that user 6 performs SIC on user 4 and 7 on subcarrier 2,  then we have. SIC Ordering and scheduling  Algorithm example output
	 	 %------------------------------------
	 	 \begin{figure*}[t]
	 	 	\centering
	 	 	\includegraphics[width=.75\textwidth]{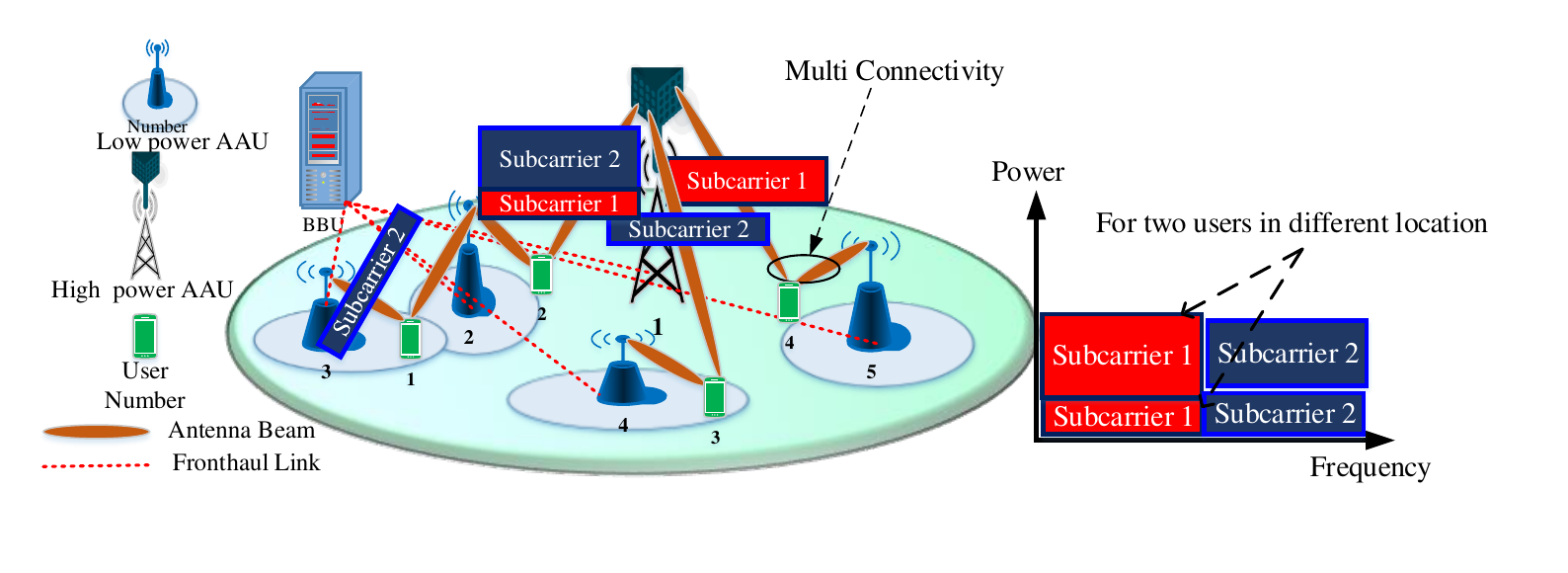}
	 	 	\vspace{-2em}
	 	 	\caption{Schematic presentation of our designed MISO NOMA-enabled C-RAN and an example of the proposed  SIC ordering algorithm.}
	 	 	\label{PSMA}
	 	 \end{figure*}
 	 %\vspace{-1.8em}
	 	 %------------------------------------------- 
	 %^^^^^^^^^^^^^^^^^^^^^^^^^^^^^^^^^^^^^^^^^^^^^^^^^^6
%	%\footnote{In order to reduce the complexity for evaluate performance of the proposed method we consider single cell.}
%	with $U$ users, $N$ subcarriers, and $C$ subcarriers. \textcolor{black}{A typical illustration of the considered system model is depicted in  Fig. \ref{PSMA}. As seen, with assuming that  user $4$ is determined to perform SIC on subcarrier $2$, it performs SIC to cancel the signal of user $2$. To this end, user $4$ uses MPA  to detect and decode own symbols.}
%	In the considered system model,  \textcolor{black}{ $\mathcal{K}$ is the set of users\footnote{\textcolor{black}{We assume that all users are  randomly deployed inside the cell.}}}, $\mathcal{N}$ is the set of subcarriers, and $\mathcal{C}$ is the set of subcarriers. 
\\\indent
	In addition, $s_{k,n}^{f}$ is the signal\footnote{We assume its power is normalized to one, i.e., $\Bbb{E}\big\{|s_{k,n}^{f}|^{2}\big\}=1$.} of user $k$ over subcarrier $n$ from AAU $f$, and $\bold{w}_{k,n,f}=[w_{k,n,f}^{m}]\in\mathcal{C}^{M\times 1}$ is the vector of beamforming variables that is designed by AAU $ f $ for user $k$ over subcarrier $n$. 
%	\begin{figure*}
		\begin{table}[!ht]
		\renewcommand{\arraystretch}{1.15}
		\centering
		\caption{Table of the main notations}
		\label{Notations}
			\textcolor{blue}{
		\begin{tabular}{ c|| l }	
			\hline
			{\textbf{Notation}}& {\textbf{Description}}\\\hline
		  \multicolumn{2}{l}{~~~~~~~~~~~~\textbf{Notations/parameters}}
			\\\hline
			$\mathcal{F}/F/f$&Set/number/index of all AAUs  in the network
			\\\hline
			 $ \mathcal{K}/K/k $&Set/number/index of all users in the network
			\\ \hline
			$\mathcal{N}/N/n$ &Set/number/index of  the shared subcarriers in each AAU\\ \hline
			${M}/m$ &Number/index  of  antennas in each AAU \\ \hline
			$P_{\text{max}}^{f}$&Maximum allowable transmit power of AAU $f$
			\\\hline
			$P_{\text{Total}}$&Total consumed power
			\\\hline
				$R_{\max}^{f}$&Maximum capacity of fronthaul link AAU $ f $
			\\\hline
			$\bold{h}_{k,n,f}$ &Real channel coefficient between user $k$ \\&and AAU $ f $  on  subcarrier $n$
						\\\hline
			$\tilde{\bold{h}}_{k,n,f}$ &Estimated channel coefficient between user $k$ \\&and AAU $ f $  on  subcarrier $n$
			\\\hline
			$\boldsymbol\epsilon_{k,n,f}$ & Channel estimation error for user $k$ and AAU $ f $  \\& on  subcarrier $n$
			\\\hline
			$ \delta_{k,n}^{f} $ & Channel uncertainty radius  for user $k$ and AAU $ f $  \\& on  subcarrier $n$
			\\\hline
		%	$\gamma_{k,n,f}$& Received SINR of user $k$ on subcarrier $n$ from AAU $f$\\\hline
			$L_{n}^{f}$& Maximum reused number of each subcarrier $n$ at  AAU $f$ \\\hline
			%\\&
			%	$\mu_{m,c}$& Variable that serves as an upper bound for the achievable\\& rate. \\\hline
			${\sigma_{k,n,f}^{2}}$ &Variance of noise at user $k$ on subcarrier $n$ from AAU $f$\\\hline
			%	${\sigma_{u,c}}$ & Noise at user $k$ on subcarrier $c$ (PSMA) \\\hline
			%	${y_{k,n}}$ &Received signal at user $k$ on subcarrier $n$\\\hline
			$s_{k,n}^{f}$ &Transmit signal at user $k$ on subcarrier $n$ from AAU $f$\\\hline	
			$r_{k,n}^{f}$ &Achieved rate of user $k$ on subcarrier $n$ from AAU $f$\\\hline	
			$\varphi_{f}$ & Preprocessing weight of each fronthual link of AAU $f$ \\\hline
			$\beta$ & Drain efficiency of the power amplifier
			\\\hline
			 \multicolumn{2}{l}{~~~~~~~~~~~~\textbf{ Optimization Variables}}
				\\\hline ${\rho_{k,n}^{f}}$ & Binary subcarrier assignment variables,
			%	between user $k$ and
			%	\\& AAU $ f $ on subcarrier $n$,
			equals to 1 means \\ & that user $k$  is scheduled to AAU $f$ and subcarrier $n$, \\& otherwise, it is 0\\\hline
							$\xi_{i,n}^{k}$& Binary SIC ordering variable, which $\xi_{i,n}^{k}=1$, if user $k$ \\& decodes the signal of user $i$ 	on subcarrier $n$, else $\xi_{i,n}^{k}=0$
							\\\hline
							${\bold{w}_{k,n,f}}$ & Beamforming vector from AAU $ f $ to user $k$ \\&on subcarrier $n$
							\\ \hline
		\end{tabular}
		}
	\end{table}	
%\end{figure*}
%	Therefore, the transmitted signal from the AAU $f$ is obtained by
%	\begin{align}
%	\bold{s}_f=\sum_{k\in\mathcal{K}}\sum_{ n \in \mathcal{N}}\bold{w}_{k,n,f}x_{k,n}^{f}.
%\end{align}
%To determine which one of subcarrier is used for which user/users,
%-----------
Let us define a joint subcarrier and user associations binary variable\footnote{Herein, we call it  the scheduling variable.},  $\rho_{k,n}^{f}$, if subcarrier $n$ is assigned to user $k$ that is served by $f$, $\rho_{k,n}^{f}=1$, otherwise, $\rho_{k,n}^{f}=0$.
% \eqref{scheduling}:
%\begin{figure*}
% \begin{align}\label{scheduling}
% \nonumber&\rho_{k,n}^{f}=
% \\&\left\{
% \begin{array}{ll}
% 1&\text{Subcarrier $n$ is assigned to user $k$ that is served by $f$,}\\
% 0&\text{Otherwise}
% \end{array}
% \right.
% \nonumber
% \end{align}
%With regards \eqref{scheduling}, 
	We introduce our scheduling policy in terms of subcarrier 
	%i.e., resource block (subcarrier)\footnote{In this paper, subcarrier and subcarrier are the same.} 
	assignment technique and connectivity of users to AAUs as follows:
%\paragraph*{\textit{Scheduling Policy}}
\\$\bullet$~\textit{Subcarrier Assignment as Multiple Access Technique:}
%\end{figure*}
\textcolor{blue}{By exploiting NOMA, each subcarrier $n$ can be assigned to at most $L_{n}^{f}$ users in AAU $f$  which is ensured by}
\begin{align}
\sum_{ k\in \mathcal{K}}\rho_{k,n}^{f}\le L_{n}^{f},\,\forall n\in\mathcal{N}, f\in\mathcal{F}.
\end{align}
\\$\bullet$~\textit{User Association as Connectivity Technique:}
In general, user association refers to find an algorithm to assign users to the radio stations. %, in contrast, to assign subcarriers to users.
We propose a novel user association policy  where each user can be configured  to receive its data on different subcarriers from different AAUs. We call it  multi-connectivity technique which is different from the coordinated multipoint  technologies\footnote{Because it dose not require synchronization between different AAUs.} \cite{Multi-Connectivity}. 
%By this way, we can   introduce the wireless  link diversity which improves the throughput\footnote{This technique also can improve the reliability of the network.} of the network,
 %more especially for cell-edge users with poor channel conditions.
%Based on this, each user can be connected to different AAUs on different subcarriers. 
\textcolor{blue}{
Therefore, each user on each subcarrier can be connected to at most one AAU   which is ensured by the following constraint:	
	%-------------------------------
%	 In first, each user should be assigned to at most one AAU which is imposed on the following constraint:
%	%. This assumption can be taken into account with the following  constraint:
%\begin{align}
%%\label{constraintoneuseronetransmitereachusereachsubcarrier1}%\nonumber
%%& \rho_{f,n,k}+\rho_{f',n',k}\le 1,\forall i\neq j, i,j \in \mathcal{I},%\\&
%%n_i \in \mathcal{N}_i, n_j \in \mathcal{N}_j,f_i,f_j\in\mathcal{F},\forall{k}\in\mathcal{K}_{\text{Total}},\\\label{constraintoneuseronetransmitereachusereachsubcarrier2}%\nonumber
%& \rho_{k,n}^{f}+\sum_{f'\in\mathcal{F}\backslash{\{f\}}}\rho_{k,n'}^{f'}\le 1,\forall f\in \mathcal{F},n,n'\in \mathcal{N},{k}\in\mathcal{K}.
%\end{align}  
%In second, each user is allowed to receive its information on different subcarriers from different AAUs, i.e., multiple transmission which is guaranteed by the following constraint:
%------------------------------------
	\begin{align}\label{User_Assosiaction_2}
% \rho_{k,n}^{f}+\sum_{f'\in\mathcal{F}\backslash{\{f\}}}\rho_{k,n}^{f'}\le 1,
%	\forall f\in \mathcal{F},n\in \mathcal{N},{k}\in\mathcal{K}.~~~
{ \sum_{f\in\mathcal{F}}\rho_{k,n}^{f}\le 1,
		\forall n\in \mathcal{N},{k}\in\mathcal{K}}.
\end{align}
%\begin{align}
%\label{User_Assosiaction_2Single}
%& \rho_{k,n}^{f}+\sum_{f'\in\mathcal{F}\backslash{\{f\}}}\rho_{k,n}^{f'}\le 1,
%\forall f\in \mathcal{F},n\in \mathcal{N},{k}\in\mathcal{K}.
%\end{align}
%In the following, we discuss  the channel.
}
%\\$\bullet$~\textit{\textcolor{blue}{Imperfect CSI:}}

\textcolor{blue}{ Let  $\bold{h}_{k,n,f}\in\mathcal{C}^{M\times 1}$ be the channel coefficient between user $k$ and AAU $f$ on subcarrier $n$.
%and $\gamma_{k,n,f}$ be the received signal interference to noise ratio (SINR) of user $k$ in AAU $f$ on subcarrier $n$.
%Motivated  to design an algorithm that is applicable in practice, 
%Due to uncertainty in the channel gain,
Following channel uncertainty, i.e., imperfect CSI model, we assume that the global CSI is not known  because of estimation errors and/or feedback delays \cite{Robust,Good_Robust}.} 
   Therefore,  the  real channel gain is given as follows\cite{Robust}:
\begin{align}
\bold{h}_{k,n,f}=\tilde{\bold{h}}_{k,n,f}+\boldsymbol{\epsilon}_{k,n,f},\,\forall k,n,f,
\end{align}
where $\tilde{\bold{h}}_{k,n,f}$ denotes the estimated channel gain and $\boldsymbol{\epsilon}_{k,n,f}$ indicates the error of estimation which lies in a bounded spherical set as given by  $\|\boldsymbol{\epsilon}_{k,n,f}\|_{2}^{2}\le \delta_{k,n}^{f}$, where %$\boldsymbol{\Delta}_{k,n,f}=[\Delta_{k,n,f}^{m}]$ 
	$ \delta_{k,n}^{f} $ is the channel uncertainty radius and is assumed be a small constant \cite{Good_Robust, CSI_Robust_TVT}. %denotes the radius of the uncertainty region
	In the other words, %for the real channel vector,
	 we have $\bold{h}_{k,n,f}\in\mathcal{H}_{k,n,f}$, where $\mathcal{H}_{k,n,f}$ is as follows:
\begin{align}
\mathcal{H}_{k,n,f}\triangleq\Big\{\tilde{\bold{h}}_{k,n,f}+\boldsymbol{\epsilon}_{k,n,f}~\big|~\|\boldsymbol{\epsilon}_{k,n,f}\|_{2}^{2}\le \delta_{k,n}^{f}\Big\}.
\end{align}
The indispensable part of NOMA is the SIC algorithm which is applied in the receiver side to handle the NOMA interference. Since in NOMA, the SIC ordering has a key impact on the received signal to interference plus noise ratio (SINR), and the performance of NOMA for cell-edge or cell-central users \cite{Zakeri}, we devise a new SIC ordering method as follows:
	\subsubsection{Proposed SIC Ordering Algorithm}\label{SIC_Section}
%$\rho_{6,2}^{1}=1~~~$ $\rho_{4,2}^{1}=1$~~$\rho_{7,2}^{4}=1$	$\xi_{6,2}^{1(4)}=0$ $	\xi_{4,2}^{4(7)}=0$$
%	$\xi_{7,2}^{1(4)}=1$
In contrast to sorting  users  based on the channel condition to perform SIC, we introduce
		a new binary variable as $\xi_{i,n}^{k}$,   where $\xi_{i,n}^{k}=1$ if user $k$ 
		%who is associated to  AAU $f$  performs SIC on
		decodes the signal of user $i$ on the assigned subcarrier $n$ (assuming both users $k$ and $i$ are multiplexed on subcarrier $n$), and otherwise, $\xi_{i,n}^{k}=0$. Note that users $ i $ and $ k $ can be connected to different AAUs.  \textcolor{blue} {It is worth noting  that the} traditional SIC ordering is based on the channel power gain, channel gain \cite{Ding_2,Ding_3,FNOMA}, and normalized noise power \cite{Joint_Optimal_20}.
		% which aims at ordering of users based on optimization problem.. 
		%The extension of this algorithm is in our previous work \cite{8786250}.
		\\$\bullet$~\textit{ Worst-Case Data Rate:}
\textcolor{blue}{The achievable  rate of user $k$ on subcarrier $n$ and AAU $f$
	with channel $\bold{h}_{k,n,f}$ is obtained by \eqref{rate_for} shown at the top of next page.} %\cite{Zakeri}
	\begin{figure*}
	\begin{align}
%	\tiny
 \label{rate_for}
	r_{k,n}^{f}=
	%\rho_{k,n}\log_2\left(1+\frac{|\bold{h}_{k,n}^{\dagger}\bold{w}_{k,n}|^2}{\sigma_{k,n,f}^2}\right)+
	\log_2\left(1+\frac{\rho_{k,n}^{f}|\bold{h}_{k,n,f}^{\dagger}\bold{w}_{k,n,f}|^2}	
	{
		\sum\limits_{f'\in\mathcal{F}}
		\sum\limits_{i\in \mathcal{K}\backslash \{k\}}\rho_{i,n}^{f'}\cdot(1-\xi_{i,n}^{k})\cdot |\bold{h}_{k,n,f'}^{\dagger}\bold{w}_{i,n,f'}|^2
		%	\sum_{i\in \mathcal{K}\backslash k}\rho_{i,n}^{f}(1-\xi_{k,n}) |\bold{h}_{k,n}^{\dagger}\bold{w}_{i,n,f'}|^2
		+\sigma_{k,n,f}^2}\right),
	\end{align}
	\hrule
\end{figure*}
The worst-case data rate of user $k$ over the uncertainty set can be formulated as
	\begin{align}\label{worst_case_Rate}
	r_{k}=\sum_{f'\in \mathcal{F}}\sum_{ n \in \mathcal{N}} \min_{\Big\{\bold{h}_{k,n,f}\in\mathcal{H}_{k,n,f},\forall f
		% ~\bold{h}_{k,n,f'}\in\mathcal{H}_{k,n,f'},\forall f'\in\mathcal{F}\backslash f
		\Big\}}r_{k,n}^{f'},~~\forall k\in\mathcal{K}.
	%\Big\{\bold{h}_{k,n,f}\in\mathcal{H}_{k,n,f}, ~\bold{h}_{k,n,f'}\in\mathcal{H}_{k,n,f'},\forall f'\in\mathcal{F}\backslash f\Big\}
	\end{align}
	$\bullet~$\textit{Successful Decoding Constraints as SIC Constraints:}
%	\subsubsection{Successfully Interference Cancellation and Fronthaul Link Capacity Constraints}
To ensure that user $k$ can successfully cancel the signal of user $i$, i.e., user $k$ is determined to perform SIC on subcarrier $n$ which means  $\xi_{i,n}^{k}=1$, we consider the following three constraints should be satisfied, simultaneously.
%\textit{a}:
% $ \blacklozenge $
\begin{enumerate}
    \item 
  SIC ordering variable can be $ 1 $, when both users $ k $ and $ i $ are multiplexed on subcarrier $ n $ which is ensured by
\begin{align}\label{SIC_Sub_Con}
\xi_{i,n}^{k}\le \sum_{f\in\mathcal{F}}\rho_{k,n}^{f}\cdot\sum_{f'\in\mathcal{F}}\rho_{i,n}^{f'}, \forall n\in\mathcal{N},~i\neq k, k,i\in\mathcal{K}.
\end{align}
%\\\textit{b}:
%$ \blacklozenge $ \textcolor{blue}{
\item
	Ensuring that  one of the user $ k $ or user $ i $ performs SIC over the multiplexed subcarrier $ n $ by
\begin{align}
\label{SIC_Con_1}
\xi_{i,n}^{k}+\xi_{k,n}^{i}\le 1,~\forall n\in\mathcal{N}, k,i\in\mathcal{K}, k\neq i.
\end{align}
%\textit{c}:
%$ \blacklozenge $ 
\item
Successful decoding constraint, which ensures that signal of user $ i $ (connected to $ f' $) on subcarrier $n$ is detected and cancelled by user $ k $ (connected to $ f $) in the worst condition (based on \eqref{rate_for}), is given by \eqref{SIC_Ordeing} shown at the top of the next page.
\begin{figure*}
%	\tiny
\begin{align}\label{SIC_Ordeing}
%\nonumber
&\underbrace{\xi_{i,n}^{k}}_{\text{A}} \underbrace{\max_{\Big\{\bold{h}_{i,n,f}\in\mathcal{H}_{i,n,f},\forall f\Big\}}}_{\text{B}}
	\underbrace{\log_2\left(1+\frac{\rho_{i,n}^{f'}|\bold{h}_{i,n,f'}^{\dagger}\bold{w}_{i,n,f'}|^2}	
	{
		\sum_{f\in\mathcal{F}}
		\sum_{k'\in \mathcal{K}\backslash \{i\}}\rho_{k',n}^{f}(1-\xi_{k',n}^{i}) |\bold{h}_{i,n,f}^{\dagger}\bold{w}_{k',n,f}|^2
		%	\sum_{i\in \mathcal{K}\backslash k}\rho_{i,n}^{f}(1-\xi_{k,n}) |\bold{h}_{k,n}^{\dagger}\bold{w}_{i,n,f'}|^2
		+\sigma_{i,n,f}^2}\right)}_{\text{C}}
%\log(1+\frac{|\bold{h}_{k,n,f}^{\dagger}\bold{w}_{k,n}|^{2}}{\sum_{i\in \mathcal{K}\backslash k}|\bold{h}_{k,n,f}^{\dagger}\bold{w}_{i,n,f'}|^{2}+\sigma_{k,n,f}^{2}})
\\&\nonumber
\le
\underbrace{\xi_{i,n}^{k}}_{\text{A}} \underbrace{\min_{\bold{h}_{k,n,f}\in\mathcal{H}_{k,n,f},\forall f}}_{\text{D}}
\underbrace{\log_2\left(1+\frac{\rho_{i,n}^{f'}|\bold{h}_{k,n,f'}^{\dagger}\bold{w}_{i,n,f'}|^2}	
	{
		\sum_{f\in\mathcal{F}}
		\sum_{k'\in \mathcal{K}\backslash \{i\}}\rho_{k',n}^{f}(1-\xi_{k',n}^{i}) |\bold{h}_{k,n,f}^{\dagger}\bold{w}_{k',n,f}|^2
		%	\sum_{i\in \mathcal{K}\backslash k}\rho_{i,n}^{f}(1-\xi_{k,n}) |\bold{h}_{k,n}^{\dagger}\bold{w}_{i,n,f'}|^2
		+\sigma_{k,n,f}^2}\right)}_{\text{E}}.
%\\&~\forall i,k\in\mathcal{K},f,f'\in\mathcal{F},n\in\mathcal{N},
\nonumber
\end{align}
\hrule
\end{figure*}
In \eqref{SIC_Ordeing}, part A
%We product $\xi_{i,n}^{k}$ to each sides of \eqref{SIC_Ordeing}, to
ensures the constraint holds for user $k$ that is determined to perform SIC to decode and remove user $i$'s signal where both of them are multiplexed on subcarrier $n$.
%(controlled by $\xi_{i,n}^{k}$), otherwise, it needs to be $0$ ($ 0 \le 0 $). 
Parts B and D assure that the constraint holds for the worst-case estimation of CSI.  To be clear, assume an example where we have two parameters as $ a\in[a_{\min},a_{\max}] $ and $ b\in[b_{\min},b_{\max}] $, and to ensure  an inequality $a\le b$ for all/worst-cases, obviously, it occurs for $a_{\max}\le b_{\min}$.
%In fact, this ensures that the condition for successful cancellation for all cases of channel gains is provided. 
Part  C is the obtained rate of user $i$ and part E is the rate of user $i$ achieved by user $k$ \cite{MISO_Optimal,Ding_2,Ding_3}. It should be noted that  \eqref{SIC_Ordeing} is sufficient (but not necessary) for SIC.
\end{enumerate}
%-----------------------------
\subsubsection{Fronthaul Link Capacity Constraints}
Since the bandwidth of  fronthaul links are  limited, we introduce a new link capacity constraint as 
 \begin{align}\label{Fronthaul_Con}
 \sum_{k\in\mathcal{K}}\sum_{n\in\mathcal{N}}\varphi_{f} \cdot r_{k,n}^{f}\le R_{\max}^{f}, \forall f\in\mathcal{F},
 \end{align}
 where $\varphi_{f}$ and $R_{\max}^{f}, \forall f\in\mathcal{F}$ are the preprocessing weight related to the fronthaul transmission technologies and maximum available transmission capacity  of AAU $f$, respectively.

%Due to uncertainty of CSI, we have uncertainty value for $r_{k,n}$. The total rate of user $k$ is $r_{k}=\sum_{ n \in \mathcal{N}} r_{k,n}$. To ensure each user achieves a target data rate from the QoS assurance perspective, we need to guarantee this for  a worse case, i.e., $\tilde{r}_{k,n}=\min_{\bold{h}_{k,n,f}^{\dagger}}r_{k,n}$. Based on this, we introduce the following constraint:
%\begin{align}\label{R_min Con}
%\sum_{ n\in\mathcal{N}}\{\tilde{r}_{k,n}\}=\sum_{ n\in\mathcal{N}}\big\{\min_{\bold{h}_{k,n,f}^{\dagger}}r_{k,n}\big\}\ge R_{k}^{\min}, ~\forall k\in\mathcal{K},
%\end{align} 
%where $R_{k}^{\min}$ is the minimum requirement data rate. This constraint guarantees  a target data rate for each user in any situation on channel condition.
%	Note that in traditional SIC ordering rate of users is obtained as 
%	\begin{align}%\nonumber
%	&r_{k,n}=\rho_{k,n}^{f}\log_2\left(1+\frac{|\bold{h}_{k,n,f}^{\dagger}\bold{w}_{k,n,f}|^2}{\sum_{i\in \mathcal{K}\backslash k,|\bold{h}_{i,n}|_{2}^{2}\ge |\bold{h}_{k,n,f}|_{2}^{2}}\rho_{i,n} |\bold{h}_{k,n,f}^{\dagger}\bold{w}_{i,n,f'}|^2+\sigma_{k,n,f}^2}\right).
%	\end{align}
%	\subsubsection{Succeful Decoding Constraint}
%	To ensure user $k$ decode the message of user $k'$ on a common subcarrier $n$, successfully, we introduce following constraint:
%	\begin{align}
%	\log(1+\frac{\rho_{k,n}^{f}|\bold{h}_{k,n,f}^{\dagger}\bold{w}_{k,n,f}|^2}{\sum_{i\in \mathcal{K}\backslash k}\rho_{i,n} |\bold{h}_{k,n,f}^{\dagger}\bold{w}_{i,n,f'}|^2+\sigma_{k,n,f}^2})
%	\end{align}
%\subsubsection{Objective Function}
\subsection{{Objective Function and Problem Formulation}}\label{problems formulation}
In this section, the considered objective function and problem formulation are introduced.
\subsubsection{Objective Function}
Considering the worst-case channel uncertainties, the main goal of the optimization problem is to maximize the worst-case
%The main goal of the optimization problem  is to maximize the worst
 global EE (GEE) of the system. %\textcolor{blue}{EE is defined the utilized power in a unit time for transmitting a bit in each Hz,}
GEE is defined as the ratio of
 	the global achievable sum rate to the total consumed
 	power \cite{GEE},
 	and the worst-case of GEE is obtained with considering the worst-case CSI in our model. To formulate the worst-case EE, 
we need the worst-case throughput of the system and the total consumed power. %Based on \eqref{rate_for}, 
The total worst-case data rate can be calculated by
\begin{align}%\nonumber
\label{Total_data_rate}
R_{\text{Total}}^{\text{Worst}}=\sum_{k\in \mathcal{K}} r_{k}.
%\sum_{ k\in \mathcal{K}}\sum_{ n \in \mathcal{N}}r_{k,n}^{f}.
\end{align} 
%where $r_{k}$ is given by \eqref{worst_case_Rate}.
The total power 
%The main goal of the optimization problem  is to minimize the system energy. The total energy of the system includes 
%the static and dynamic terms are sum of 
%total 
% consumption  of BS $P_{\text{BS}}$ and backhAAUl transport devices 
%that is used for traffic traversing \cite{power_model}. Therefore, the total
  consumption of the system is obtained by \cite{power_model}
\begin{align}\label{Total_Power}
P_{\text{Total}}=\sum_{f\in \mathcal{F}}\underbrace{\sum_{k\in\mathcal{K}}\sum_{n\in\mathcal{N}}\frac{1}{\beta}\|\bold{w}_{k,n,f}\|_2^2}_{P_{\text{TX}}^{f}}+P_{\text{Static}},%_{P_{\text{BS}}},
%+\underbrace{\beta\sum_{ k\in \mathcal{K}}\sum_{ n \in \mathcal{N}}r_{k,n}^{f}}_{P_{\text{BackhAAUl}}},
\end{align}
where $0<\beta<1$ is the drain efficiency of the power amplifier
%$\beta$ is the power consumed %by the backhAAUl devices for transmitting
% per unit of the data traffic, 
 and $P_{\text{Static}}$ is the static term of power consumption which is obtained as
%\begin{equation}\label{Static Power}
$P_{\text{Static}}=\sum_{f\in \mathcal{F}}P_{\text{Static}}^{f}$, where  $ P_{\text{Static}}^{f} $ is the static term of power which  is given by
\begin{equation}
P_{\text{Static}}^{f}=\left\{
\begin{array}{ll}
P_{\text{Hardware}}^{f}& 0< P_{\text{TX}}^{f}\le P_{\max}^{f},\\
P_{\text{Sleep}}^{f}&P_{\text{TX}}^{f}=0,
\end{array}
\right.
\end{equation}
where $P_{\text{Sleep}}^{f}$ is the consumed power at BBU in the sleep mode of AAU $ f $ and  $P_{\text{Hardware}}^{f}=C_{\text{Circuit}}\times M$ is the power that is used by hardware at AAU $ f $ in the transmission mode, and $C_\text{Circuit}$ is the consumed circuit power constant that is used for the signal processing functions at AAU which includes the power dissipation in the
filtering, frequency synthesizer, digital-to-analog converter, etc. Therefore, the worst-case EE of the system is calculated by 
\begin{align}
\eta_{\text{EE}}^{\text{Worst}}=\frac{R_{\text{Total}}^{\text{Worst}}}{P_{\text{Total}}}, % \min_{{\bold{H}}}
\end{align}  
where $R_{\text{Total}}^{\text{Worst}}$ and $P_{\text{Total}}$ are given by \eqref{Total_data_rate} and \eqref{Total_Power}, respectively.
% Alternative express of EE can be introduce as spectral efficiency denoted by $\eta_{\text{SE}}$ over of the total consumed power as $\frac{\eta_{\text{SE}}}{P_{\text{Total}}}.$
%Note that we consider $P_{\text{BackhAAUl}}$, since the generated data traffic by the radio access network has major effect on power consuming of the transport network. Hence, optimizing data traffic in a power sensitive network is essential.
	\subsubsection{Problem Formulation} Based on these definitions and assumptions, our main aim is to maximize the worst case of EE considering beamforming and SIC ordering constraints. The optimization problem is mathematically formulated as follows:
		\begin{subequations}\label{sic_order_main}
			\begin{align}\label{Reeq8a11}
			\max_{\bold{W},\boldsymbol{\xi},\boldsymbol{\rho}}\;&~\eta_{\text{EE}}^{\text{Worst}}
			\\\label{power_cons1}
			\text{s.t.}~~&
		%	\begin{array}
			\text{C}_{1}:\sum_{k\in \mathcal{K}}\sum_{n\in \mathcal{K}}\rho_{k,n}^{f}\|\bold{w}_{k,n,f}\|_{2}^{2}\le P_{\text{max}}^{f},\forall f\in\mathcal{F},
			\\&\label{NOMA_Con}
				\text{C}_{2}:\sum_{ k\in \mathcal{K}}\rho_{k,n}^{f}\le L_{n}^{f},\,\forall n\in\mathcal{N}, f\in\mathcal{F},
			\\&\label{BS_Ass_Con}	
			\text{C}_{3}:{ \sum_{f\in\mathcal{F}}\rho_{k,n}^{f}\le 1,
				\forall n\in \mathcal{N},{k}\in\mathcal{K}},
			%
			%\rho_{k,n}^{f}+\sum_{f'\in\mathcal{F}\backslash{\{f\}}}\rho_{k,n}^{f'}\le 1,\forall f\in \mathcal{F},n\in \mathcal{N},{k}\in\mathcal{K},
			%\\&\label{}
		%	\text{C}_{5}:\eqref{SIC_Ordering_Con}
				\\&\label{Reeq8i1}
			\text{C}_{4}:
			\xi_{i,n}^{k},\rho_{k,n}^{f}\in
			\begin{Bmatrix}
			0,
			1
			\end{Bmatrix},\,\,\forall k\in \mathcal{K}, n\in\mathcal{N}, \forall f\in\mathcal{F},
						\\&\label{SICs.11}
			%\min_{|\mathbf{h}_{m,n}|\in \omega_{m,n}}
			 \nonumber
		  \eqref{SIC_Sub_Con}, \eqref{SIC_Con_1}, \eqref{SIC_Ordeing}, \eqref{Fronthaul_Con}, % \eqref{SIC_Or_Num},  \eqref{Com_Delay}
		%	\end{array}
			%\\&\label{ratemin}r_m\ge r_{min},\,\,\,\forall m\in\mathcal{K}.
			\end{align}
		\end{subequations}
where $\bold{W}=\left[{w}_{k,n,f}^{m}\right]$,~$\boldsymbol{\rho}=\left[\rho_{k,n}^{f}\right]$,~ and $\boldsymbol{\xi}=\left[\xi_{i,n}^{k}\right]$.
 %,~ and $\bold{L}=[L_{n}^{f}], \forall f\in\mathcal{F}, k,i\in\mathcal{K},n\in\mathcal{N}$. 
 \textcolor{blue}{
 Constraint \eqref{power_cons1} indicates the maximum available power budget and constraint \eqref{NOMA_Con} verifies that each subcarrier $n$ in each AAU $f$ can be utilized no more than $L_{n}^{f}$ times, recognized as NOMA constraint. 
	% $ \text{C}_{5} $
 Constraint \eqref{BS_Ass_Con} ensures that each user on each subcarrier is served only by one AAU, \eqref{Reeq8i1} stands for binary variables, and  \eqref{SIC_Sub_Con} indicates that the SIC on each subcarrier between users which are multiplexed on that subcarrier. Moreover,~\eqref{SIC_Con_1} indicates that for any two users only one of them can perform SIC on another one and 	\eqref{SIC_Ordeing}
	  ensures that user $k$ decodes the message of user $i$ on  the assigned subcarrier $n$, successfully \cite{Ding_2, Fairness}.
	 Finally,  constraint \eqref{Fronthaul_Con} is the link capacity restriction of fronthaul links.
	 }
    %where $r^S_{k,n}$ and $r^I_{q,n}$ are obtained from \eqref{ratespdnoma} and \eqref{SIC2}, respectively.
	 % Constraint \eqref{min_rate} ensures that each user achieves a minimum data rate according to its QoS requirements, and the
	  %Constraint \eqref{sic_con} indicates that at most one user can perform SIC on each subcarrier.
	  	 % 
	  %	 \vspace{-1.2em}
	\section{Proposed Solution Methods}\label{solutions}
The optimization problem in \eqref{sic_order_main} is a non-convex mixed integer non-linear programming (MINLP) which is complicated to solve. We propose two different solution algorithms which are discussed in the following.
\subsection{Algorithm 1: One Step Solution}
In this section, we explain our proposed one-step solution, i.e., all variables are obtained without using alternating approach.
First, let us  define matrices $\bold{W}_{k,n,f}$ and $\tilde{\bold{H}}_{k,n,f}$ with size $ M\times M $ as 
 $\bold{W}_{k,n,f}\triangleq\bold{w}_{k,n,f}\bold{w}_{k,n,f}^{\dagger}$
	and $\tilde{\bold{H}}_{k,n,f}\triangleq\tilde{\bold{h}}_{k,n,f}\tilde{\bold{h}}_{k,n,f}^{\dagger}$, respectively. To this end, we rewrite $\|\bold{h}_{k,n,f}^{\dagger}\bold{w}_{k,n,f}\|^2$ as follows:
	\begin{align}
\nonumber	&\|\bold{h}_{k,n,f}^{\dagger}\bold{w}_{k,n,f}\|^2=
	\nonumber\\&
	\bold{w}_{k,n,f}^{\dagger}(\tilde{\bold{h}}_{k,n,f}+\boldsymbol{\epsilon}_{k,n,f})^{\dagger}(\tilde{\bold{h}}_{k,n,f}+\boldsymbol{\epsilon}_{k,n,f})\bold{w}_{k,n,f}
	\nonumber\\ &=\bold{w}_{k,n,f}^{\dagger}(\tilde{\bold{h}}_{k,n,f}+\boldsymbol{\Delta}_{k,n,f})\bold{w}_{k,n,f}
	\nonumber\\&=\text{Tr}[(\tilde{\bold{H}}_{k,n,f}+\boldsymbol{\Delta}_{k,n,f})\bold{W}_{k,n,f}],
\end{align}
	where $\boldsymbol{\Delta}_{k,n,f}=\tilde{\bold{h}}_{k,n,f}\boldsymbol{\epsilon}_{k,n,f}^{\dagger}+\boldsymbol{\epsilon}_{k,n,f}\tilde{\bold{h}}_{k,n,f}+\boldsymbol{\epsilon}_{k,n,f}\boldsymbol{\epsilon}_{k,n,f}^{\dagger}$ is a norm-bounded matrix which satisfies the following region: 
\begin{align}
%r4rr\nonumber
\label{Delta}
		\|\boldsymbol{\Delta}_{k,n,f}\|&\leq 	\|\tilde{\bold{h}}_{k,n,f}\boldsymbol{\epsilon}_{k,n,f}^{\dagger}\|+\|\boldsymbol{\epsilon}_{k,n,f}\tilde{\bold{h}}_{k,n,f}^{\dagger}\|+\|\boldsymbol{\epsilon}_{k,n,f}\boldsymbol{\epsilon}^{\dagger}_{k,n}\|,\nonumber\\
		&\leq 	\|\tilde{\bold{h}}_{k,n,f}\|\|\boldsymbol{\epsilon}_{k,n,f}^{\dagger}\|+\|\boldsymbol{\epsilon}_{k,n,f}\|\|\tilde{\bold{h}}_{k,n,f}^{\dagger}\|
		\nonumber\\&
		+\|\boldsymbol{\epsilon}_{k,n,f}\|\|\boldsymbol{\epsilon}^{\dagger}_{k,n}\|
	\nonumber\\
		&
		=(\delta^{f}_{k,n})^2+2\delta_{k,n}^f\|\tilde{\bold{h}}_{k,n,f}\|\triangleq e_{k,n}^{f}.
	\end{align}  
	Therefore, equation \eqref{rate_for} can be rewritten as \eqref{Rate_Error} shown at the top of the next page.
		\begin{figure*}
\begin{align}\label{Rate_Error}\tiny
	&{r}^{f}_{k,n}=
	\log_2\left(1+\frac{\rho_{k,n}^{f}\text{Tr}[(\tilde{\bold{H}}_{k,n,f}+\boldsymbol{\Delta}_{k,n,f})\bold{W}_{k,n,f}]}{\sum_{f'\in\mathcal{F}}
		\sum_{i\in \mathcal{K}\backslash \{k\}}\rho_{i,n}^{f'}(1-\xi_{i,n}^{k}) \text{Tr}[(\tilde{\bold{H}}_{k,n,f'}+\boldsymbol{\Delta}_{k,n,f'})\bold{W}_{i,n,f'}]+\sigma_{k,n,f}^2}\right).
\end{align}
	\hrule
\end{figure*}
	\textcolor{blue}{
	Now, we  aim at maximizing the worst-case data rate \eqref{Rate_Error}. 
	Since the log function is a monotonic function, finding the worst-case would be done over the SINR in (\ref{Rate_Error}). This can be obtained by minimizing \eqref{Rate_Error} over $\Delta_{k,n,f}$ and 
	$\Delta_{k,n,f'} $ where indexes $f$ and $f'$  can be the same, i.e., when we calculate the intra-cell interference.
	One conservative method to find the minimum of the SINR is minimizing the numerator and maximizing the denominator of SINR in (\ref{Rate_Error}) with respect to norm-bounded matrices \cite{Proof_Worst}. Note that by this method, we provide a strictly bounded robust solution (SBRS) \cite{Proof_Worst}.
	Motivated by this idea,~the lower bound of SINR in (\ref{Rate_Error}) subject to \eqref{Delta} can be obtained by solving
% 	\begin{align}\label{Robust}
% 	&	\min_{	\|\boldsymbol{\Delta}_{k,n,f}\|\leq e_{k,n}^{f}} \text{Tr}[(\tilde{\bold{H}}_{k,n,f}+\boldsymbol{\Delta}_{k,n,f})\bold{W}_{k,n,f}] \nonumber\\ &- \sum_{f'\in\mathcal{F}}
% 		\sum_{i\in \mathcal{K}\backslash \{k\}}\max_{\{	\|\boldsymbol{\Delta}_{k,n,f'}\|\leq e_{k,n}^{f'}\}}\rho_{i,n}^{f'}(1-\xi_{i,n}^{k})\nonumber \\&
% 		\text{Tr}[(\tilde{\bold{H}}_{k,n,f'}+\boldsymbol{\Delta}_{k,n,f'})\bold{W}_{i,n,f'}]+\sigma_{k,n,f}^2.  
% 	\end{align}
%	To solve \eqref{Robust},
	the following optimization problems:
	%should be solved 
\begin{align}
&\min_{	\|\boldsymbol{\Delta}_{k,n,f}\|\leq e_{k,n,f}} \text{Tr}[(\tilde{\bold{H}}_{k,n,f}+\boldsymbol{\Delta}_{k,n,f})\bold{W}_{k,n,f}], \label{minimumterm}
\\&
\label{21max}
\max_{	\|\boldsymbol{\Delta}_{k,n,f'}\|\leq e_{k,n,f'}}
		\sum_{f'\in\mathcal{F}}
		\sum_{i\in \mathcal{K}\backslash \{k\}}\text{Tr}[(\tilde{\bold{H}}_{k,n,f'}+\boldsymbol{\Delta}_{k,n,f'})\bold{W}_{i,n,f'}].
\end{align}
	Here, we apply the Lagrangian-based method to find  the optimal solution of \eqref{minimumterm} and \eqref{21max} for the given the beamforming matrices \cite{Proof_Worst}.
	}
	\textcolor{blue}{
	\begin{proposition} \label{Pro_SBRS}
		For the given $\bold{W}_{k,n,f}$, the minimum and maximum norm-bounded matrices of \eqref{minimumterm} and \eqref{21max} given as follows,~respectively, for all $ k,i,n,f,f'$:
\begin{align}
	&\Delta^{f,\min}_{k,n}=-e_{k,n,f} \frac{\bold{W}^{\dagger}_{k,n,f}}{\|\bold{W}_{k,n,f}\|}\label{min_error},
	~\forall k, n, f, %\in\mathcal{K},\forall n\in\mathcal{N},f\in\mathcal{K},
	\\
	&\Delta^{f',\max}_{k,n}=e_{k,n,f'} \frac{\bold{W}^{\dagger}_{i,n,f'}}{\|\bold{W}_{i,n,f'}\|},~\forall k,i\in\mathcal{K}, n, f, k\neq i\label{max_error}.
\end{align}
	\end{proposition} 
\begin{proof}
	Please see Appendix A.
\end{proof}
The following remark provides the exact solution for the worst-case of SINR.
}
\textcolor{blue}{
\begin{remark} \label{EXRS}
The exact worst-case of the SINR can be obtained by using the fractional programming.
%However, we can exactly determine  the worst case of SINR.
The 
minimization of the SINR over a bounded error is a fractional program as follows:
\begin{align}
\label{Frac}
&\min_{\|\boldsymbol{\Delta}_{k,n,f}\|\le e_{k,n,f}}
\\&\nonumber
    \frac{\text{Tr}[(\tilde{\bold{H}}_{k,n,f}+\boldsymbol{\Delta}_{k,n,f})\bold{W}_{k,n,f}]}{\sum_{f'\in\mathcal{F}}
		\sum_{i\in \mathcal{K}\backslash \{k\}} \text{Tr}[(\tilde{\bold{H}}_{k,n,f'}+\boldsymbol{\Delta}_{k,n,f'})\bold{W}_{i,n,f'}]+\sigma_{k,n,f}^2}.
\end{align}
The solution of \eqref{Frac} can be obtained using the idea of fractional programming.
%for a given $\bold{W}_{k,n,f},\forall k,n,f$. %\cite{Frac_Ed}.
This method provides a numerical solution for the worst-case $\Delta_{k,n,f}$. 
%We are interested in the solution to handle the original problem and obtain optimization variables such as $\bold{W}_{k,n,f}$.
Hence, by idea of fractional programming, we restate \eqref{Frac} in parametric form as:
\begin{align}
    \min_{\|\boldsymbol{\Delta}_{k,n,f}\|\le e_{k,n,f}}
   F(\boldsymbol{\Delta}_{k,n,f},\nu),
\end{align}
where 
\begin{align}
    &   F(\boldsymbol{\Delta}_{k,n,f},\nu) \triangleq
    {\text{Tr}[(\tilde{\bold{H}}_{k,n,f}+\boldsymbol{\Delta}_{k,n,f})\bold{W}_{k,n,f}]}
    \\&\nonumber
    -\nu\Big({\sum_{f'\in\mathcal{F}}
		\sum_{i\in \mathcal{K}\backslash \{k\}} \text{Tr}[(\tilde{\bold{H}}_{k,n,f'}+\boldsymbol{\Delta}_{k,n,f'})\bold{W}_{i,n,f'}]+\sigma_{k,n,f}^2}\Big),
\end{align}
and 
$\nu$ is an auxiliary variable updated in each iteration of Dinkelbach algorithm by
\begin{align}
\label{Q_Update}
&\nu=
\\&\nonumber
\frac{\text{Tr}[(\tilde{\bold{H}}_{k,n,f}+\boldsymbol{\Delta}_{k,n,f}^{*})\bold{W}_{k,n,f}]}{{\sum_{f'\in\mathcal{F}}
		\sum_{i\in \mathcal{K}\backslash \{k\}} \text{Tr}[(\tilde{\bold{H}}_{k,n,f'}+\boldsymbol{\Delta}_{k,n,f'}^{*})\bold{W}_{i,n,f'}]+\sigma_{k,n,f}^2}},
		\end{align} 
		where $\boldsymbol{\Delta}_{k,n,f}^{*}$ is obtained by \eqref{Delta_Max} which will be explained next. In each iteration of the Dinkelbach algorithm, we  need to obtain $\text{argmin}_{\|\boldsymbol{\Delta}_{k,n,f}\|\le e_{k,n,f}} F(\boldsymbol{\Delta}_{k,n,f},q)$. We adopt the Lagrangian method. To this end, we define the Lagrangian function as follows:
\begin{align}
\label{Lag}
    \mathcal{L}(\boldsymbol{\Delta}_{k,n,f})=F(\boldsymbol{\Delta}_{k,n,f},\nu)+\Omega (\|\boldsymbol{\Delta}_{k,n,f}\|- e_{k,n,f}),
\end{align}
where $\Omega$ is the Lagrangian multiplier. 
Taking the derivations of \eqref{Lag} with respect to $\boldsymbol{\Delta}_{k,n,f}$ and $\Omega$, and setting their values to zero, we obtain $\boldsymbol{\Delta}_{k,n,f}^{*}$ as follows:
\begin{align}
    \label{Delta_Max}
    \boldsymbol{\Delta}_{k,n,f}^{*}=
    \text{argmin}_{\mathcal{L}(\boldsymbol{\Delta}_{k,n,f})}.
\end{align}
The update procedure in Dinkelbach algorithm is proceed until the convergence condition is met. We call this method as exact robust solution (ExRS). It is worthwhile mentioning that, since the objective function in \eqref{Frac} is pseudo-linear, we can solve the max (best-case) using the same approach.
\end{remark}
   }
By substituting \eqref{min_error} and \eqref{max_error} into the data rate \eqref{Rate_Error}, we obtain \eqref{Tota_DR_24},
\begin{figure*}
\begin{align}
	\label{Tota_DR_24}
	\tilde{r}_{k,n}^{f}=\log_2\bigg(1+\frac{\rho_{k,n}^{f}\big(\text{Tr}[\tilde{\bold{H}}_{k,n,f}\bold{W}_{k,n,f}]-e_{k,n}^{f}{\|\bold{W}_{k,n,f}\|}\big)}{\sum_{f'\in\mathcal{F}}\sum_{i\in \mathcal{K}\backslash \{k\}}\rho_{i,n}^{f'}(1-\xi_{i,n}^{k})\big\{\text{Tr}[\tilde{\bold{H}}_{k,n,f'}\bold{W}_{i,n,f'}]+e_{k,n}^{f'}{\|\bold{W}_{i,n,f'}\|}\big\}+\sigma_{k,n,f}^2}\bigg). 
\end{align}
\hrule
\end{figure*}
where $ \tilde{r}_{k,n}^{f}$ is the lower bound on the worst-case data rate. Therefore, the total data rate for the worst-case is given by 
\begin{align}\label{Tota_DR}
R_{\text{Total}}^{\text{Worst}}=\sum_{f \in \mathcal{F}}\sum_{ k\in \mathcal{K}}\sum_{ n \in \mathcal{N}}\tilde{r}_{k,n}^{f}.
\end{align}

Moreover, we restate the constraint \eqref{SIC_Ordeing} by using Proposition \ref{Pro_SBRS} as \eqref{SIC_C} shown at top of the next page which may not be tight.
%\textcolor{blue}{
\begin{figure*}
\textcolor{blue}{
	\begin{align}\label{SIC_C}
& \log_2\bigg(1+\frac{\xi_{i,n}^{k}\rho_{i,n}^{f'}\big\{\text{Tr}[\tilde{\bold{H}}_{i,n,f'}\bold{W}_{i,n,f'}]+e_{i,n}^{f'}{\|\bold{W}_{i,n,f'}\|}\big\}}{\sum_{f\in\mathcal{F}}\sum_{k'\in \mathcal{K}\backslash \{i\}}\rho_{k',n}^{f}(1-\xi_{k',n}^{i})\big\{\text{Tr}[\tilde{\bold{H}}_{i,n,f}\bold{W}_{k',n,f}]-e_{i,n}^{f}{\|\bold{W}_{k',n,f}\|}\big\}+\sigma_{i,n,f'}^2}\bigg)
\\&\nonumber
\le
 \log_2\bigg(1+\frac{\xi_{i,n}^{k}\rho_{i,n}^{f'}\big\{\text{Tr}[\tilde{\bold{H}}_{k,n,f'}\bold{W}_{i,n,f'}]-e_{k,n}^{f'}{\|\bold{W}_{i,n,f'}\|}\big\}}{\sum_{f\in\mathcal{F}}\sum_{k'\in \mathcal{K}\backslash \{i\}}\rho_{k',n}^{f}(1-\xi_{k',n}^{i})\{\text{Tr}[\tilde{\bold{H}}_{k,n,f}\bold{W}_{k',n,f}]+e_{k,n}^{f}{\|\bold{W}_{k',n,f}\|}\}+\sigma_{k,n,f}^2}\bigg).
%\\&~\forall i,k\in\mathcal{K},f,f'\in\mathcal{F},n\in\mathcal{N}.
\end{align} 
}
\hrule
\end{figure*}
In \eqref{SIC_C}, the left hand side is obtained for the maximum while the right hand side is obtained for the minimum. 
%}
	 We also note that the term $\{\xi_{i,n}^{k}\cdot\rho_{i,n}^{f'}\}$ in \eqref{Tota_DR_24} is non-convex. To tackle this issue,~we define a new variable as the multiplication of two binary variables $x_{i,n}^{k,f'}=\xi_{i,n}^{k}\cdot\rho_{i,n}^{f'}$, which indicates the joint SIC ordering and scheduling variables. 
	 \textcolor{blue}{Then, we adopt a linearizion technique and add the following constraints to the optimization problem \cite{Linear19,Linear2}:
	\begin{align}
	 \text{C}_{7}:~&\xi_{i,n}^{k} \geq x_{i,n}^{k,f},~\forall k,i\in\mathcal{K}, n\in\mathcal{N}, f\in\mathcal{F},\\
	 \text{C}_{8}:~&\rho_{i,n}^{f} \geq x_{i,n}^{k,f},~\forall k,i\in\mathcal{K}, n\in\mathcal{N},f\in\mathcal{F},\\
	 \text{C}_{9}:~& \rho_{i,n}^{f}+\xi_{i,n}^{k}-1 \leq x_{i,n}^{k,f},\forall k,i\in\mathcal{K},n\in\mathcal{N},f\in\mathcal{F}.
\end{align}
}
	Furthermore, in order to handle the non-convex integer variables in problem \eqref{sic_order_main},~we rewrite the constraint \eqref{Reeq8i1} as the intersection of the following regions:
	\begin{align}\small
\nonumber	 &
	 \text{C}_{10}: 0 \leq\xi_{i,n}^{k}\leq 1,~\text{C}_{11}: \sum_{ i,k\in \mathcal{K}}\sum_{ n \in \mathcal{N}}\bigg( \xi_{i,n}^{k}-(\xi_{i,n}^{k})^{2}\bigg) \leq 0,
	\nonumber  \\     
	& \text{C}_{12}: 0 \leq\rho_{k,n}^{f}\leq 1,~
	 \nonumber\\&\text{C}_{13}:\sum_{f\in \mathcal{F}}\sum_{ k\in \mathcal{K}}\sum_{ n \in \mathcal{N}} \bigg( \rho_{k,n}^{f}-(\rho_{k,n}^{f})^{2}\bigg)  \leq 0,
	\nonumber \\
	& \text{C}_{14}:  0 \leq x_{i,n}^{k,f}\leq 1,
	\nonumber\\&\text{C}_{15}:\sum_{f\in \mathcal{F}}\sum_{ i,k\in \mathcal{K}}\sum_{ n \in \mathcal{N}}  \bigg(x_{i,n}^{k,f}-(x_{i,n}^{k,f})^{2}\bigg) \leq 0.
	\nonumber
\end{align} 
\textcolor{blue}{Now,~we rewrite the optimization problem as:
	 	\begin{subequations}\label{sic_order_change}
	 	\begin{align}\label{Reeq23}
	 	\max_{\bold{W},\boldsymbol{\xi},\boldsymbol{\rho}}\;&~\eta_{\text{EE}}^{\text{Worst}}
	\\\label{power_cons.11}
	 	\text{s.t.}~
	 	&\text{C}_{1}:\sum_{k\in \mathcal{K}}\sum_{n\in \mathcal{K}}\text{Tr}[\rho_{k,n}^{f}\bold{W}_{k,n,f}]\leq P^{f}_{\text{max} },\\
%	&\text{C}_{3}:\|\boldsymbol{\epsilon}_{k,n,f}\|_{2}^{2}\le \boldsymbol{\Delta}_{k,n,f},\\
 	&\text{C}_{2}:\sum_{ k\in \mathcal{K}}\rho_{k,n}^{f}\le L_{n}^{f},\\
 	%\\&\label{sic_con1}\text{C}_{5}:\sum_{k\in\mathcal{K}}x_{k,n}^{f}\leq 1,\,\forall n\in\mathcal{N},\\
 	&\text{C}_{3}:\bold{W}_{k,n,f}\succeq \bold{0},\\
 & \text{C}_{4}:\text{Rank}(\bold{W}_{k,n,f})\leq 1,\\
 %+\lambda \Bigg(\sum_{f\in \mathcal{F}}\sum_{i,k\in \mathcal{K}}\sum_{ n \in \mathcal{N}} \xi_{i,n}^{k}-(\xi_{i,n}^{k})^{2} +\sum_{f\in \mathcal{F}}\sum_{ k\in \mathcal{K}}\sum_{ n \in \mathcal{N}} \rho_{k,n}^{f}\nonumber\\&-(\rho_{k,n}^{f})^{2}+\sum_{f\in \mathcal{F}}\sum_{ k\in \mathcal{K}}\sum_{ n \in \mathcal{N}} x_{k,n}^{f}-(x_{k,n}^{f})^{2}\Bigg)\\
 &	
	 	\text{C}_{7}-\text{C}_{15},
	 	\\&
	 	\eqref{User_Assosiaction_2}, \eqref{SIC_Sub_Con}, \eqref{SIC_Con_1}, \eqref{Fronthaul_Con}, \eqref{SIC_C}.
	 	\end{align}
 \end{subequations}
It can be concluded from the optimization problem (\ref{sic_order_change}),~the product term of $\rho_{k,n}^{f}\bold{W}_{k,n,f}$ is an obstacle for solving the optimization problem.~Let us define two new auxiliary variables as follows:
\begin{align}\nonumber
   \tilde{\bold{W}}_{k,n,f}\triangleq\rho_{k,n}^{f}\bold{W}_{k,n,f},~~
{\bold{W'}}_{i,k,n,f}\triangleq x_{i,n}^{k,f}\bold{W}_{i,n,f}.
\end{align}
Also, we employ the big-M method \cite{Big,FNOMA} to circumvent this difficulty.~In particular, we impose the following additional constraints to make it convex as follows:
\begin{align}
\text{C}_{16}:~&\tilde{\bold{W}}_{k,n,f}\preceq  P_{\text{max}}^{f} \bold{I}_{M}\rho_{k,n}^{f},\\
\text{C}_{17}:~&\tilde{\bold{W}}_{k,n,f}\preceq \bold{W}_{k,n,f},~\text{C}_{18}:~\tilde{\bold{W}}_{k,n,f}\succeq \bold{0},\\
\text{C}_{19}:~&\tilde{\bold{W}}_{k,n,f}\succeq \bold{W}_{k,n,f}-(1-\rho_{k,n}^{f}) P_{\text{max} }^f \bold{I}_{M},\\
\text{C}_{20}:~&{\bold{W'}}_{i,k,n,f}\preceq  P_{\text{max}}^f \bold{I}_{M}~x_{i,n}^{k,f},\\
\text{C}_{21}:~&{\bold{W}'}_{i,k,n,f}\preceq \bold{W}_{k,n,f},~\text{C}_{22}:~{\bold{W}'}_{i,k,n,f}\succeq \bold{0},\\
\text{C}_{23}:~&{\bold{W}'}_{i,k,n,f}\succeq \bold{W}_{k,n,f}-(1-x_{i,n}^{k,f}) P_{\text{max} }^f \bold{I}_{M}.
\end{align}
}
 %\textcolor{blue}{
 	\textcolor{blue}{
 	%Besides the above linearizations,  the objective function, i.e., 
 	The worst-case data rate \eqref{Tota_DR} and constraint \eqref{SIC_C} are still non-convex. To handle these and facilitate the solution, \eqref{Tota_DR} can be rewritten as  
 \begin{align}\label{Worst_Rate_1}
R_{\text{Tota}}^{\text{Worst}}=\log_2 \prod_{f\in \mathcal{F}}\prod_{k\in \mathcal{K}}\prod_{n\in \mathcal{N}} \frac{\psi_{f,k,n}}{\phi_{f,k,n}},
 \end{align}
 where 
 \begin{align}
 \nonumber
 &\psi_{f,k,n}=\text{Tr}[\tilde{\bold{H}}_{k,n,f}\tilde{\bold{W}}_{k,n,f}]-e_{k,n}^{f}\big(\mathcal{A}(\tilde{\bold{W}}_{k,n,f})\big)+ \phi_{f,k,n},
% +\sum_{f' \in \mathcal{F}}\sum_{i\in \mathcal{K}\backslash \{k\}}\text{Tr}[\tilde{\bold{H}}_{k,n,f'}\tilde{\bold{W}}_{i,n,f'}]-\nonumber\\&\text{Tr}[\tilde{\bold{H}}_{k,n,f'}{\bold{W}'}_{i,k,n,f'}]+
% e_{k,n}^{f'}\Big(A(\bold{W}_{i,n,f'})-B(\bold{W}'_{i,k,n,f'})\Big)+\sigma^2_{k,n,f},
\\ &
\nonumber
 \phi_{f,k,n}=
 \nonumber\\ &\sum_{f' \in \mathcal{F}}\sum_{i\in \mathcal{K}\backslash \{k\}}\text{Tr}[\tilde{\bold{H}}_{k,n,f'}\tilde{\bold{W}}_{i,n,f'}]
 -\text{Tr}[\tilde{\bold{H}}_{k,n,f'}{\bold{W}'}_{i,k,n,f'}]\nonumber
\\&+
 e_{k,n}^{f'}\Big(\mathcal{A}(\tilde{\bold{W}}_{i,n,f'})-\mathcal{B}(\bold{W}'_{i,k,n,f'})\Big)+\sigma^2_{k,n,f},
 \end{align}
 where we define
 \begin{align}
 \label{aux40}
 &\mathcal{A}(\tilde{\bold{W}}_{k,n,f})\triangleq\|\rho_{k,n}^{f}\bold{W}_{k,n,f}\|={\|\tilde{\bold{W}}_{k,n,f}\|},
% \\&\label{aux41}
% B(\bold{W}_{i,n,f'})\triangleq\|\rho_{i,n}^{f'}\bold{W}_{i,n,f'}\|={\|\tilde{\bold{W}}_{i,n,f'}\|},
 \\&\label{aux42}\mathcal{B}(\bold{W}'_{i,k,n,f'})\triangleq\|x_{i,n}^{k,f'}\bold{W}_{i,n,f'}\|=\|{\bold{W}'}_{i,k,n,f'}\|.
 \end{align}
 }
 However, the total data rate \eqref{Worst_Rate_1} is still non-convex. Let us define 
 \begin{align}\label{slack}
 	\psi_{f,k,n}\triangleq \exp ({a_{f,k,n}}),~\phi_{f,k,n}\triangleq \exp({b_{f,k,n}}),
 \end{align}
 where $a_{f,k,n}$ and $b_{f,k,n}$ are  slack  variables.~Moreover, the lower bound of slack variables have the following form \cite{DerrikIOT}
 \begin{align}\label{bound}
 \exp ({a_{f,k,n}})\geq \sigma^{2}_{k,n,f},~\exp({b_{f,k,n}})\geq\sigma^{2}_{k,n,f}.
 \end{align}
 %Similarly, we handle the SIC constraint via introducing auxiliary variables $c_{f,k,n}$,~$d_{f,k,n}$,~$\chi_{f,k,n}$, and $\nu_{f,k,n}$.
%\textcolor{blue}{
%\textcolor{blue}{
%\\\indent
Now, by substituting (\ref{slack}) into the objective function, we can obtain
 \begin{align}
 %\eta_{\text{EE}}^{\text{Worst}}=
 \max_{\bold{W},\bold{W'},\tilde{\bold{W}},\boldsymbol{\xi},\boldsymbol{\rho},\bold{x}}\;&~ \frac{\log_2 \prod_{f}\prod_{k}\prod_{n} \frac{\psi_{f,k,n}}{\phi_{f,k,n}}}{ P_\text{Total}(\bold{W})}\nonumber\\=\max_{\bold{W},\bold{W'},\tilde{\bold{W}},\boldsymbol{\xi},\boldsymbol{\rho},\bold{x},\bold{a},\bold{b}}&\frac{\log_2 \prod_{f}\prod_{k}\prod_{n} \exp ({a_{f,k,n}}-b_{f,k,n})} { P_\text{Total}(\bold{W})}\nonumber\\\nonumber=\max_{\bold{W},\bold{W'},\tilde{\bold{W}},\boldsymbol{\xi},\boldsymbol{\rho},\bold{x},\bold{a},\bold{b}}&\frac{
 \sum \limits_{f\in\mathcal{F}}
\sum\limits_{ k\in \mathcal{K}}\sum \limits_{ n \in \mathcal{N}}({a_{f,k,n}}-b_{f,k,n})\log_2(e)}{ P_\text{Total}(\bold{W})},
 \end{align}
 where $\bold{a}=[a_{f,k,n}]$ and $\bold{b}=[b_{f,k,n}]$ are the collection vectors of slack variables.
 \textcolor{blue}{ Also, $P_{\text{Total}}(\bold{W})=\sum_{f\in \mathcal{F}}\sum_{k\in\mathcal{K}}\sum_{n\in\mathcal{N}}\frac{1}{\beta}\text{Tr}[\bold{W}_{k,n,f}]+P_{\text{Static}}$.
 For notation simplicity, let define $ \Xi\triangleq\left[\bold{W'},\tilde{\bold{W}},{\bold{W}},\boldsymbol{\xi},\boldsymbol{\rho},\bold{x},\bold{a},\bold{b}\right] $ as the collection of the optimization variables. 
 It is worthwhile to mention that by this 
%treating the original non-convex objective function, 
 method, we can easily rewrite the non-convex constraint \eqref{Fronthaul_Con}  as follows:
 \begin{align}
 \label{Fronthualu1}
     \sum_{k\in\mathcal{K}}\sum_{n\in\mathcal{N}}\varphi_{f}(a_{f,k,n}-b_{f,k,n})\le R_{\max}^{f}, \forall f\in\mathcal{F}.
 \end{align}
 }
 Now, the optimization problem at hand can be mathematically formulated as
 \begin{subequations}\label{DC_Problem}
 	\begin{align}\label{Reeq28}
 	&\max_{\Xi}
 	\\&
% 		\bold{W'},\tilde{\bold{W}},{\bold{W}},\boldsymbol{\xi},\boldsymbol{\rho},\bold{x},\bold{a},\bold{b}}
 	\frac{\sum\limits_{f \in \mathcal{F}}\sum\limits_{ k\in \mathcal{K}}\sum\limits_{ n \in \mathcal{N}}({a_{f,k,n}}-b_{f,k,n})\log_2(e)
 	}
 	{ P_\text{Total}(\bold{W})}+h(\boldsymbol{\xi},\boldsymbol{\rho},\bold{x},\lambda)
	\\\label{PC}
	\text{s.t.}~~
	&\text{C}_{1}:\sum_{k\in \mathcal{K}}\sum_{n\in \mathcal{K}}\text{Tr}[\tilde{\bold{W}}_{k,n,f}]\leq P^{f}_{\text{max}},~ %\forall f \in \mathcal{F},
\\
 %	&\text{C}_{3}:\|\boldsymbol{\epsilon}_{k,n,f}\|_{2}^{2}\le {\delta}_{k,n,f},\\
	&\text{C}_{2}:\sum_{ k\in \mathcal{K}}\rho_{k,n}^{f}\le L_{n}^{f},\,
	%\forall n\in\mathcal{N},\forall f\in\mathcal{F},
	\\ 	&\text{C}_{3}:\bold{W}_{k,n,f}\succeq \bold{0},\\
	& \text{C}_{4}:\text{Rank}(\bold{W}_{k,n,f})\leq 1,\\
	& \exp ({a_{f,k,n}})\geq \sigma^{2}_{k,n,f},~\exp({b_{f,k,n}})\geq\sigma^{2}_{k,n,f},\\
	&a_{f,k,n}\geq b_{f,k,n}\label{A_bound},\\
	\label{Phi_Bound}
	&\psi_{f,k,n}\geq \exp(a_{f,k,n}),~\phi_{f,k,n}\leq \exp(b_{f,k,n}),\\
 	& \text{C}_{7}-\text{C}_{10},\text{C}_{12},\text{C}_{14},\text{C}_{16}-\text{C}_{23},
 	\\& 
 	\eqref{User_Assosiaction_2}, \eqref{SIC_Sub_Con}, \eqref{SIC_Con_1}, \eqref{Fronthaul_Con}, \eqref{SIC_C},
	\end{align}
 \end{subequations}
 \textcolor{blue}{
   where $ h(\boldsymbol{\xi},\boldsymbol{\rho},\bold{x},\lambda)\triangleq	-f(\boldsymbol{\xi},\boldsymbol{\rho},\bold{x},\lambda)+g(\boldsymbol{\xi},\boldsymbol{\rho},\bold{x},\lambda) $ is the penalty function. Furthermore,
   $f(\boldsymbol{\xi},\boldsymbol{\rho},\bold{x},\lambda)$~and $g(\boldsymbol{\xi},\boldsymbol{\rho},\bold{x},\lambda)$ are defined as
%}
   	\begin{align}
   	\tiny
  \nonumber &f(\boldsymbol{\xi},\boldsymbol{\rho},\bold{x},\lambda)\triangleq
  \nonumber \\&
   \lambda \Bigg(\sum_{ i,k\in \mathcal{K}}\sum_{ n \in \mathcal{N}} \xi_{i,n}^{k}+\sum_{f\in \mathcal{F}}\sum_{ k\in \mathcal{K}}\sum_{ n \in \mathcal{N}} \rho_{k,n}^{f}+\sum_{f\in \mathcal{F}}\sum_{ i,k\in \mathcal{K}}\sum_{ n \in \mathcal{N}} x_{i,n}^{k,f}\Bigg),
   \\
  &g(\boldsymbol{\xi},\boldsymbol{\rho},\bold{x},\lambda)\triangleq
   \lambda \Bigg(\sum_{ i,k\in \mathcal{K}}\sum_{ n \in \mathcal{N}} (\xi_{i,n}^{k})^{2}
  \nonumber\\&+\sum_{f\in \mathcal{F}}\sum_{ k\in \mathcal{K}}\sum_{ n \in \mathcal{N}} (\rho_{k,n}^{f})^{2}+\sum_{f\in \mathcal{F}}\sum_{ i,k\in \mathcal{K}}\sum_{ n \in \mathcal{N}} (x_{i,n}^{k,f})^{2}\Bigg),
   \end{align}
   	respectively.~It is worth mentioning that $\lambda$ indicates a penalty factor to penalize the objective function for any $\xi_{i,n}^{k}$,~$\rho_{k,n}^{f}$,~and $x_{k,n}^{f}$ which are not binary (i.e., their values are in $[0,1]$) \cite{Ata_TWC}. However, the precise $0$ and $1$  are not always available. In this case, we {round} their values to the nearest integer values. 
   The following proposition provides a mathematical analysis of the penalty factor. 
   	}
 \textcolor{blue}{
 	  \begin{proposition}\label{Pro_Binary}
   	For sufficiently large value of $\lambda$, we have~$\min_{\lambda}\max_{\bold{W},\boldsymbol{\xi},\boldsymbol{\rho},\bold{x}}\mathcal{L}(\bold{W},\boldsymbol{\xi},\boldsymbol{\rho},\bold{x})=\max_{\bold{W},\boldsymbol{\xi},\boldsymbol{\rho},\bold{x}}\min_{\lambda}\mathcal{L}(\bold{W},\boldsymbol{\xi},\boldsymbol{\rho},\bold{x})$\footnote{Note that $\mathcal{L}(\bold{W},\boldsymbol{\xi},\boldsymbol{\rho},\bold{x})$ is the objective function in (\ref{sic_order_change}).}.~In other words, the optimization problem \eqref{sic_order_change} is equivalent to \eqref{DC_Problem} where both problems obtain the same values.
   \end{proposition}
%}
\begin{proof}
	Please see Appendix B.
\end{proof}
   Problem \eqref{DC_Problem} is non-convex due to the constraints \eqref{Phi_Bound} and \eqref{SIC_C} and $ g(\boldsymbol{\xi},\boldsymbol{\rho},\bold{x},\lambda) $ in the objective function. 
    In order to convert it into a convex one, we employ MM approach where a surrogate function is approximated by the first  order Taylor approximation.~Therefore,~we use the following inequalities:
   \begin{align}\label{g_small}
   \nonumber
   g(\boldsymbol{\xi},\boldsymbol{\rho},\bold{x},\lambda)&\leq g(\boldsymbol{\xi}^{t-1},\boldsymbol{\rho}^{t-1},\bold{x}^{t-1},\lambda)\\&+\nabla_{\boldsymbol{\xi}}g^{T}(\boldsymbol{\xi}^{t-1},\boldsymbol{\rho}^{t-1},\bold{x}^{t-1},\lambda)(\boldsymbol{\xi}-\boldsymbol{\xi}^{t-1})\nonumber\\&+\nabla_{\boldsymbol{\rho}}g^{T}(\boldsymbol{\xi}^{t-1},\boldsymbol{\rho}^{t-1},\bold{x}^{t-1},\lambda)(\boldsymbol{\rho}-\boldsymbol{\rho}^{t-1})\nonumber\\&+\nabla_{\bold{x}}g^{T}(\boldsymbol{\xi}^{t-1},\boldsymbol{\rho}^{t-1},\bold{x}^{t-1},\lambda)(\bold{x}-\bold{x}^{t-1})\nonumber\\&\triangleq\tilde{g}(\boldsymbol{\xi},\boldsymbol{\rho},\bold{x},\lambda),
   %\\
   \end{align}
   }
%   where
  \textcolor{blue}{
  	  \begin{align}
  	  %\label{A_approx}
%  	  &A(\bold{W}'_{k,n,f})\geq A(\bold{W}'_{k,n,f})^{t-1}+\text{Tr}\Big(\nabla_{\bold{W}'_{k,n,f}}A(\bold{W}'_{k,n,f})^{t-1}
%  	  \nonumber\\&
%  	  \big(A(\bold{W}'_{k,n,f})-A(\bold{W}'_{k,n,f})^{t-1}\big)\Big),
%  	  %\triangleq \tilde{A}(\bold{W}'_{k,n,f})
%  	  \\
     &\mathcal{A}(\tilde{\bold{W}}_{k,n,f})\leq \mathcal{A}(\tilde{\bold{W}}_{k,n,f})^{t-1}\nonumber\\&+\text{Tr}\Big(\nabla_{\tilde{\bold{W}}_{k,n,f}}\mathcal{A}(\tilde{\bold{W}}_{k,n,f})^{t-1}
   [\mathcal{A}(\tilde{\bold{W}}_{k,n,f})^{t}-\mathcal{A}(\tilde{\bold{W}}_{k,n,f})^{t-1}]\Big)
 \nonumber  \\&\label{A_Approx}\triangleq 
   \tilde{\mathcal{A}}(\tilde{\bold{W}}_{k,n,f}),\\
   & \mathcal{B}(\bold{W}'_{i,k,n,f'})\leq  \mathcal{B}(\bold{W}'_{i,k,n,f'})^{t-1}\nonumber\\&+\text{Tr}\Big(\nabla_{\bold{W}'_{i,k,n,f'}} \mathcal{B}(\bold{W}'_{i,k,n,f'})^{t-1}\big[ \mathcal{B}(\bold{W}'_{i,k,n,f'})^{t}\nonumber\\&- \mathcal{B}(\bold{W}'_{i,k,n,f'})^{t-1}\big]\Big)\triangleq \tilde{ \mathcal{B}}(\bold{W}'_{i,k,n,f'}),\label{C_approx}
   \\\label{b_approx}
  & \phi_{f,k,n}\leq \exp (b^{t-1}_{f,k,n})(b_{f,k,n}^t-b^{t-1}_{f,k,n}+1)\triangleq\tilde{\phi}_{f,k,n}.
   \end{align}
}
%  \textcolor{blue}{ \begin{align}\label{G_large}
%\phi_{f,k,n}\leq \exp (b^{t-1}_{f,k,n})(b_{f,k,n}-b^{t-1}_{f,k,n}+1).
%   \end{align}}
   It should be noted that the right hand sides of (\ref{g_small})-(\ref{C_approx}) are affine functions.
   \textcolor{blue}{Now the main challenge in problem \eqref{DC_Problem} is the non-convex constraint \eqref{SIC_C}, i.e.,  SIC ordering constraint. Next, we handle this constraint similar to the objective function.
   %handle this constraint using the same idea of handling the objective function.
   }
  % \begin{proposition}
%In a similar manner, we handle 
%The SIC constraint \eqref{SIC_C} can be converted to a convex one via the similar way of handling the objective function.
%introducing auxiliary variables. This  is omitted here due to lack of space.~
%\end{proposition}
%\begin{proof}
\textcolor{blue}{
   To this end, first, we rewrite  \eqref{SIC_C} as \eqref{SIC_C_u} shown at the top of the next page, by replacing the  previously defined auxiliary variables $x_{i,n}^{k,f'}=\xi_{i,n}^{k}\cdot\rho_{i,n}^{f'}$, $\tilde{\bold{W}}_{k,n,f}=\rho_{k,n}^{f}\bold{W}_{k,n,f},$
${\bold{W'}}_{i,k,n,f}=x_{i,n}^{k,f}\bold{W}_{i,n,f}$, \eqref{aux40}, and \eqref{aux42}.
%Also, in \eqref{SIC_C_u}, we use that fact that the logarithm function is a monotonic function.
\begin{figure*}
\textcolor{blue}{
	\begin{align}\label{SIC_C_u}
& \log_2\bigg(1+
\frac{\text{Tr}[\tilde{\bold{H}}_{i,n,f'}{\bold{W'}}_{i,k,n,f'}]+e_{i,n}^{f'}{\mathcal{B}({\bold{W'}}_{i,k,n,f'})}}{\sum_{f\in\mathcal{F}}\sum_{k'\in \mathcal{K}\backslash \{i\}}
%\rho_{k',n}^{f}(1-\xi_{k',n}^{i})
\Big\{\text{Tr}[\tilde{\bold{H}}_{i,n,f}\tilde{\bold{W}}_{k',n,f}]-e_{i,n}^{f}{\mathcal{A}(\tilde{\bold{W}}_{k',n,f})}
-\big\{\text{Tr}[\tilde{\bold{H}}_{i,n,f}\bold{W'}_{i,k',n,f}]-e_{i,n}^{f}\mathcal{B}(\bold{W'}_{i,k',n,f})\big\}
\Big\}
+\sigma_{i,n,f'}^2}
\bigg)
\\&\nonumber
\le
 \log_2\bigg(1+
 \frac{\text{Tr}[\tilde{\bold{H}}_{k,n,f'}\bold{W'}_{i,k,n,f'}]-e_{k,n}^{f'}{\mathcal{B}({\bold{W'}}_{i,k,n,f'})}}
 {{\sum_{f\in\mathcal{F}}\sum_{k'\in \mathcal{K}\backslash \{i\}}
 %\rho_{k',n}^{f}(1-\xi_{k',n}^{i})
 \Big\{\text{Tr}[\tilde{\bold{H}}_{k,n,f}\tilde{\bold{W}}_{k',n,f}]+e_{i,n}^{f}{\mathcal{A}(\tilde{\bold{W}}_{k',n,f})}
 -\big\{\text{Tr}[\tilde{\bold{H}}_{i,n,f}\bold{W'}_{i,k',n,f}]+e_{i,n}^{f}\mathcal{B}(\bold{W'}_{i,k',n,f})
 \big\}
 %\text{Tr}[\tilde{\bold{H}}_{k,n,f}\bold{W}_{k',n,f}]+e_{k,n}^{f}{\|\bold{W}_{k',n,f}\|}
 \Big\}+\sigma_{k,n,f}^2}}
 \bigg).
\end{align} 
}
\hrule
\end{figure*}
Now, we can rewrite \eqref{SIC_C_u} as follows:
\begin{align}
\label{SIC_Log}
    \log_2(\frac{\mathcal{D}_{i,k,n,f'}}{\mathcal{E}_{i,n,f'}})-\log_2(\frac{\mathcal{F}_{i,k,n,f'}}{\mathcal{G}_{i,k,n,f'}})\le 0,
\end{align}
  where
\begin{align}
\nonumber
 &   \mathcal{E}_{i,n,f'}=\sum_{f\in\mathcal{F}}\sum_{k'\in \mathcal{K}\backslash \{i\}}
%\rho_{k',n}^{f}(1-\xi_{k',n}^{i})
\Big\{\text{Tr}[\tilde{\bold{H}}_{i,n,f}\tilde{\bold{W}}_{k',n,f}]-e_{i,n}^{f}{\mathcal{A}(\tilde{\bold{W}}_{k',n,f})}
\\&
-\big\{\text{Tr}[\tilde{\bold{H}}_{i,n,f}\bold{W'}_{i,k',n,f}]-e_{i,n}^{f}\mathcal{B}(\bold{W'}_{i,k',n,f})\big\}
\Big\}
+\sigma_{i,n,f'}^2,\label{SICE}
\\&\nonumber
\mathcal{D}_{i,k,n,f'}=
\\&\text{Tr}[\tilde{\bold{H}}_{i,n,f'}{\bold{W'}}_{i,k,n,f'}]+e_{i,n}^{f'}{\mathcal{B}({\bold{W'}}_{i,k,n,f'})}+E_{i,n,f'},
\\\nonumber&
\label{SICG}
\mathcal{G}_{i,k,n,f'}=\sum_{f\in\mathcal{F}}\sum_{k'\in \mathcal{K}\backslash \{i\}}
 %\rho_{k',n}^{f}(1-\xi_{k',n}^{i})
\text{Tr}[\tilde{\bold{H}}_{k,n,f}\tilde{\bold{W}}_{k',n,f}]+e_{i,n}^{f}{\mathcal{A}(\tilde{\bold{W}}_{k',n,f})}
\\&
 -\big\{\text{Tr}[\tilde{\bold{H}}_{i,n,f}\bold{W'}_{i,k',n,f}]+e_{i,n}^{f}\mathcal{B}(\bold{W'}_{i,k',n,f})
 \big\}
 %\text{Tr}[\tilde{\bold{H}}_{k,n,f}\bold{W}_{k',n,f}]+e_{k,n}^{f}{\|\bold{W}_{k',n,f}\|}
 +\sigma_{k,n,f}^2,
\\\nonumber
&\mathcal{F}_{i,k,n,f'}=
\nonumber\\&\mathcal{G}_{i,k,n,f'}+\text{Tr}[\tilde{\bold{H}}_{k,n,f'}\bold{W'}_{i,k,n,f'}]-e_{k,n}^{f'}{\mathcal{B}({\bold{W'}}_{i,k,n,f'})}.
\end{align}
Hereafter, we drop the subscripts of $\mathcal{E}_{i,n,f'},~\mathcal{D}_{i,k,n,f'},~\mathcal{G}_{i,k,n,f'},~\text{and}~\mathcal{F}_{i,k,n,f'}$ for the notation simplicity and denote  them by
% the numerator and denominator of left hand side of the equation, respectively; $F$ and $G$ are the numerator and denominator of right hand of the equation. Note that, we do not bring the equations of 
$\mathcal{E}, ~\mathcal{D},~ \mathcal{G}$,~and $\mathcal{F}$, respectively.
%which means these parameter does not a scalar.
 %to   repetition avoidance and convenience. 
%are the exact mapping to the form of \eqref{SIC_C_u} as
%are the numerators and denominators as shown in
%\eqref{SIC_C_u}
 After this, similar to \eqref{slack}, we consider the following definitions:
%and corresponding constraints:
\begin{align}
\label{auxsic}
\mathcal{D}\triangleq\exp(d), \mathcal{E}\triangleq\exp(e), \mathcal{F}\triangleq\exp(f), \mathcal{G}\triangleq\exp(g),
\end{align}
where $d,~e,~f$, and $g$ are the auxiliary variables.
Substituting \eqref{auxsic} into \eqref{SIC_Log}, we obtain
\begin{align}
\begin{split}
    0&\ge \log_2(\frac{\mathcal{D}}{\mathcal{E}})-\log_2(\frac{\mathcal{F}}{\mathcal{G}}),
    %~~~\text{Constraint}~\eqref{SIC_Log} 
    \\
    &=\log_2(\exp(d-e))-\log_2(\exp(f-g)),
    \\&=\log_2e\times(d-e-(f-g)),
\end{split}
\end{align}
which has a linear form. For these auxiliary variables, we have the following constraints:
\begin{align}
\label{SICConu}
\begin{split}
     &d\ge e, ~~~f\ge g,
  \\
  & \mathcal{D}\ge\exp(d), \mathcal{E}\le\exp(e), \mathcal{F}\ge\exp(f), \mathcal{G}\le\exp(g).  
\end{split}
%\end{state}
\end{align}
Constraint \eqref{SICConu} is non-convex due to the form of $\mathcal{E}$,~$\mathcal{G}$, and constraints $\mathcal{E}\le\exp(e)$ and $\mathcal{G}\le\exp(g)$.
However, by replacing $\mathcal{A}(.)$ and $\mathcal{B}(.)$ with their linear approximations defined in \eqref{A_Approx} and \eqref{C_approx}, $\mathcal{E}$ and $\mathcal{G}$ would become affine functions denoted by $\tilde{\mathcal{E}}$ and $\tilde{\mathcal{G}}$, respectively. Finally, by employing a Taylor series expansion of $\exp(e)$ and $\exp(g)$, the constraints $\mathcal{E}\le\exp(e)$ and $\mathcal{G}\le\exp(g)$ can be transformed by
\begin{align}
\label{SIC-eg}
\begin{split}
        & \tilde{\mathcal{E}}\le\exp(\tilde{e})(e-\tilde{e}+1),
    \\& \tilde{\mathcal{G}}\le\exp(\tilde{g})(g-\tilde{g}+1),
\end{split}
\end{align}
where $\tilde{e}$ and $\tilde{g}$ are the feasible points. By considering these, the convex form of \eqref{SIC_C} is given by the following constraint:
\begin{align}\label{SIC_Con}
    \begin{split}
             &
             %\log_2e(
             d\ge e,~ f\ge g,~ \mathcal{D}\ge\exp(d),~ \eqref{SIC-eg}.
    \end{split}
\end{align}
%}
Recall that for notation simplicity, we removed the subscripts  of $d, e, f$, and $g$.
%\end{proof}
}
  % Now, constraint \eqref{SIC_C} can be replaced by \eqref{SIC_Con} 
\textcolor{blue}{As for the final step, we have to tackle non-convex constraint \eqref{SIC_Sub_Con}.~To this end,~we first handle the right hand side of \eqref{SIC_Sub_Con}, i.e., $\xi_{k,n}^{i}\le \sum_{f\in\mathcal{F}}\rho_{k,n}^{f}\cdot\sum_{f'\in\mathcal{F}}\rho_{i,n}^{f'}, k\neq i$, by using the MM approach to make a convex form that is inspired from \cite{Ata_WCL}.~It
   is straight forward to show that $z_1z_2=\frac{1}{2}[(z_1+z_2)^{2}-(z_1)^2-(z_2)^2]$, where we can define $ z_1=z_{k,n}\triangleq\sum_{f\in\mathcal{F}}\rho_{k,n}^{f}, $ and $z_2=y_{i,n}\triangleq\sum_{f'\in\mathcal{F}}\rho_{i,n}^{f'} $ to transform \eqref{SIC_Sub_Con} into a convex one. 
%   
%   $\sum_{f\in\mathcal{F}}\rho_{k,n}^{f}\cdot\sum_{f'\in\mathcal{F}}\rho_{i,n}^{f'}=\frac{1}{2}\bigg(\sum_{f \in \mathcal{F}}\rho_{k,n}^{f}+\sum_{f' \in \mathcal{F}}\rho_{i,n}^{f'}\bigg)^{2}-\frac{1}{2}\bigg(\sum_{f \in \mathcal{F}}(\rho_{k,n}^{f})^{2}-\sum_{f' \in \mathcal{F}}(\rho_{i,n}^{f'})^{2}\bigg)$.~
   Now, we rewrite the constraint in \eqref{SIC_Sub_Con} as follows:
   \begin{align}
   \nonumber
       2\xi_{k,n}^{i}\le \big(z_{k,n}+y_{i,n}\big)^2-\bigg((z_{k,n})^2+(y_{i,n})^2\bigg),\forall k,i\in\mathcal{K}, k\neq i.
   \end{align}
   The above constraint is non-convex.
 %P(\boldsymbol{\rho})-Q(\boldsymbol{\rho})$ 
 %where
%  $P(\boldsymbol{\rho})=\frac{1}{2}(z_1+z_2)^{2}$ and %${Q}(\boldsymbol{\rho})=\frac{1}{2}\big((z_1)^2+(z_2)^2\big)$. 
   Similarly, we adopt Taylor approximation for $Q_{k,n}^{i}\triangleq (z_{k,n})^2+(y_{i,n})^2$ to obtain a convex constraint.
   Therefore, after employing Taylor approximation this constraint can be written as follows:
   \begin{align}\label{7final}
         2\xi_{k,n}^{i}\le \bigg(z_{k,n}+y_{i,n}\bigg)^2-\tilde{Q}_{k,n}^{i},\forall k,i\in\mathcal{K}, k\neq i,
   \end{align}
  % $\xi_{k,n}^{i}\le[P(\boldsymbol{\rho})-\tilde{Q}(\boldsymbol{\rho})]$,
   where  $\tilde{Q}_{k,n}^{i}$ is the first order Taylor approximation for ${Q}_{k,n}^{i}$. 
    }    		We further~use the following theorem which is related to the nonlinear fractional programming.
    \textcolor{blue}{	 \begin{theorem}
    	 		A generalized fractional problem is defined as:
    	 	\begin{align}
    	 	\label{Dinki}
    	 	\max_{\mathbf{x}}~~\min_{k=1,...,K} \frac{f_{k}(\mathbf{x})}{g_{k}(\mathbf{x})},~~~\text{s.t.}~~\mathbf{x} \in \mathcal{D}.
    	 	\end{align}
    	 	 \end{theorem}
    %	 	\textcolor{blue}{
    \begin{proposition}\label{DinkiPro}
    	 	 The optimal vector $\mathbf{x}^{\ast}$ for solution of \eqref{Dinki} can be achieved if and only if \cite{Dr. Eduar} 
    	 	 \begin{equation}
    	 	 \mathbf{x}^{\ast}=\text{argmax}_{\mathbf{x}\in{\mathcal{D}}}\bigg\{\min_{k=1,...,K}[f_{k}(\mathbf{x})-\alpha^{\ast}g_{k}(\mathbf{x})]\bigg\},
    	 	 \end{equation}
    	 	 where $\alpha^{\ast}$ is obtained via the following problem:
    	 	 \begin{equation}
    	 	 	F(\alpha)=\max_{\mathbf{x} \in \mathcal{D}}\min_{k,1,...,K}[{f_{k}(\mathbf{x})-\alpha g_{k}(\mathbf{x})]=0}.
    	 	 \end{equation}
    	 	 \end{proposition}
    	 	 }
    	 	\textcolor{blue}{ By using Proposition \ref{DinkiPro},
    	 	  the optimal value of the EE  is given by
    	 	  \begin{align}
    	   \nonumber q^{\ast}\triangleq\frac{\log_2(e)\sum_{f \in \mathcal{F}}\sum_{ k\in \mathcal{K}}\sum_{ n \in \mathcal{N}}({a_{f,k,n}^{*}}-b_{f,k,n}^{*})}{P_\text{Total}(\bold{W}^{\ast})}=\\\max_{\Xi}
    %	 	  \bold{W},\boldsymbol{\xi},\boldsymbol{\rho}}
    	 	  \frac{\log_2(e)\sum_{f \in \mathcal{F}}\sum_{ k\in \mathcal{K}}\sum_{ n \in \mathcal{N}}({a_{f,k,n}}-b_{f,k,n})}{P_\text{Total}(\bold{W})}.
    	 	 	  \end{align}
    	 	  As a result,~the maximum EE, $q^{*}$ can be achieved if and only if
    	\begin{align}
 &  \nonumber \max_{\Xi}~
    % \\\nonumber&
     \log_2(e)\sum_{f \in \mathcal{F}}\sum_{ k\in \mathcal{K}}\sum_{ n \in \mathcal{N}}({a_{f,k,n}}-b_{f,k,n})-q^{\ast}P_\text{Total}(\bold{W})\nonumber\\&=\log_2(e)\sum_{f \in \mathcal{F}}\sum_{ k\in \mathcal{K}}\sum_{ n \in \mathcal{N}}({a_{f,k,n}^{*}}-b_{f,k,n}^{*}) 
   %  \nonumber \\&
     -q^{\ast}P_{\text{Total}}(\bold{W}^{\ast})
   =0.
    	 	\end{align}
    	 	}
    % 	 \begin{proof}
    % 	     For the proof, please see Appendix C. 
    % 	 \end{proof}
%    	 }
\indent
    Based on the previous steps and defining $\Xi'\triangleq[\Xi,d,e,f,g]$,~we solve the following problem instead of dealing with fractional programming in \eqref{DC_Problem}
   % instead of dealing with non-convex problem in (\ref{DC_Proble
   \begin{subequations}\label{DC_Final}
   	\textcolor{blue}{
	\begin{align}
 	\nonumber	\max_{\Xi'
% 		\bold{W},\bold{W'},\tilde{\bold{W}},\boldsymbol{\xi},\boldsymbol{\rho},\bold{x},\bold{a},\bold{b}
 	}&\sum_{f \in \mathcal{F}}\sum_{ k\in \mathcal{K}}\sum_{ n \in \mathcal{N}}({a_{f,k,n}}-b_{f,k,n})
 \\&-q\cdot P_\text{Total}(\bold{W})-f(\boldsymbol{\xi},\boldsymbol{\rho},\bold{x},\lambda)+\tilde{g}(\boldsymbol{\xi},\boldsymbol{\rho},\bold{x},\lambda)
   	\\\nonumber
   	\text{s.t.}
   	\\&\psi_{f,k,n}\geq \exp(a_{f,k,n}),
 %	\\ & 2\xi_{k,n}^{i}\le \big(z_{k,n}+y_{i,n}\big)^2-\tilde{Q}_{k,n}^{i},\forall k,i\in\mathcal{K}, k\neq i,
 %\nonumber
 \\&\text{C}_{7}-\text{C}_{10},\text{C}_{12},\text{C}_{14},\text{C}_{16}-\text{C}_{23},
    %	\nonumber\\
 \\  	&~\eqref{User_Assosiaction_2}, \eqref{SIC_Con_1}, % \eqref{SIC_Ordeing},
   	\eqref{Fronthualu1},\eqref{SIC_Con},\eqref{bound},\eqref{7final},\eqref{A_bound}.
   	\end{align}
   	}
   	   \end{subequations}
   	   \indent
\textcolor{blue}{
%Problem \eqref{DC_Final} is obtained by relaxing the rank-one constraint ${C}_{6}$.
Notice that, \eqref{DC_Final} is an standard SDP programming which
 can be solved optimally by using an off-the-shelf optimization tool, e.g., CVX. 
 %Moreover, the optimal
 %solution set of \eqref{DC_Final} is denoted by $\{\bold{W}^{*},\boldsymbol{\xi}^{*},\boldsymbol{\rho}^{*},\bold{x}^{*}\}$. In particular,
 Denote by ${\bold{W}}_{k,n,f}^\ast,\forall k,n,f,$ the optimal value of variable ${\bold{W}}_{k,n,f}$ in the solution of \eqref{DC_Final}. 
 If ${\bold{W}}_{k,n,f}^\ast$ satisfies the rank-one constraint, i.e., Rank$(\bold{W}_{k,n,f}^\ast)=1$, the optimal solution $\bold{w}_{k,n,f}^\ast$ can be obtained by using the eigenvalue decomposition (EVD) of $\bold{W}_{k,n,f}^\ast$. 
As for the final step,~we employ SDP relaxation by removing constraint ${C}_{4}$.
The problem  \eqref{DC_Final} may not yield a rank-one solution. Thus, we propose a  penalty function to the objective function to penalize it \cite{DerrikIOT,STAR}. To this end, first, we introduce the following proposition.
\begin{proposition}
\label{Pro_Eigen}
The inequality $||\mathbf{Y}||_*\triangleq\sum_{i} \sigma_i\geq ||\mathbf{Y}||_2=\underset{i}{\text{max}}\{\sigma_i\}$ holds for any given $\mathbf{Y}\in \mathbb{H}^{m\times n},m,n\in\Bbb{N}$, where $\sigma_{i}$ is the $i$-th eigenvalue value of $\mathbf{Y}$. The equality holds if and only if $\mathbf{Y}$ is rank-one.
\end{proposition}
Inspired by Proposition \ref{Pro_Eigen}, 
the equivalent form of the rank-one constraint $C_4$ can be written as $\varpi_{k,n,f}\triangleq||\mathbf{W}_{k,n,f}||_*-||\mathbf{W}_{k,n,f}||_2\leq 0,\forall k,n,f$.
 Hence, we use the penalty-based approach by integrating such a constraint into the objective function.~By introducing $\mu>1$ as a penalty factor, the objective function, by dropping the constant term $\log_2(e)$, can be rewritten as follows:
\begin{align}
\label{Rank}
&\sum_{f \in \mathcal{F}}\sum_{ k\in \mathcal{K}}\sum_{ n \in \mathcal{N}}({a_{f,k,n}}-b_{f,k,n})-q\cdot P_\text{Total}
\nonumber\\&
-f(\boldsymbol{\xi},\boldsymbol{\rho},\bold{x},\lambda)+\tilde{g}(\boldsymbol{\xi},\boldsymbol{\rho},\bold{x},\lambda) -{\mu}\times \sum_{f \in \mathcal{F}}\sum_{ k\in \mathcal{K}}\sum_{ n \in \mathcal{N}}  \varpi_{k,n,f}.
%\nonumber
\end{align}
 For a sufficiently large value of $\mu$, maximizing (\ref{Rank}) under the constraints of \eqref{DC_Final} yields a rank-one solution by ensuring a small value of $\varpi_{k,n,f}$ \cite{STAR,DerrikIOT}.
However, (\ref{Rank}) is still a non-convex function over $\mathbf{W}_{k,n,f}$. Hence, we rewrite  $\varpi_{k,n,f}$ as follows:
\begin{align}
    \tilde{\varpi}_{k,n,f}=||\mathbf{W}_{k,n,f}||_*-\tilde{\mathcal{A}}(
    \mathbf{W}_{k,n,f}),\forall k,n,f,
\end{align}
where $\tilde{\mathcal{A}}(
    \mathbf{W}_{k,n,f})$ is obtained by \eqref{A_Approx}.
%a lower bound for $U(\mathbf{W}_{k,n,f})=||\mathbf{W}_{k,n,f}||_2,~\forall k,n,f$, which is given by
% \textcolor{red}{\begin{align}
% &{\small U(\mathbf{W}_{k,n,f})\geq
% \\
% & U(\mathbf{W}_{k,n,f}^t)+\text{Tr}\bigg((\nabla_{\mathbf{W}_{k,n,f}}^HU(\mathbf{W}_{k,n,f}^{t}))(\mathbf{W}_{k,n,f}-\mathbf{W}_{k,n,f}^{t})\bigg)\triangleq\tilde{U}(\mathbf{W}_{k,n,f})}.
% \end{align}
The above equation is updated iteratively with iteration number $t$. We also update the penalty factor in each iteration as $\mu^{(t+1)}=\alpha \mu^{(t)}$, where $\alpha>1$ is a constant.
Nevertheless, the returned solution may not be rank-one. In such cases, the \underline{ Gaussian randomization} method is exploited to obtain a feasible solution. 
%}
Consequently, the optimization problem \eqref{DC_Final} can be written with same constraints by considering the following objective function:
\begin{align}\label{RankF}
&
%\max_{\Xi}
	%\bold{W},\bold{W'},\tilde{\bold{W}},\boldsymbol{\xi},\boldsymbol{\rho},\bold{x},\bold{a},\bold{b}}
	\sum_{f \in \mathcal{F}}\sum_{ k\in \mathcal{K}}\sum_{ n \in \mathcal{N}}({a_{f,k,n}}-b_{f,k,n})-q\cdot P_\text{Total}(\bold{W})
\nonumber\\
&\underbrace{\underbrace{- f(\boldsymbol{\xi},\boldsymbol{\rho},\bold{x},\lambda)+\tilde{g}(\boldsymbol{\xi},\boldsymbol{\rho},\bold{x},\lambda)}_{\text{Term 1}}\underbrace{-{\mu}\times \sum_{f \in \mathcal{F}}\sum_{ k\in \mathcal{K}}\sum_{ n \in \mathcal{N}}  \tilde{\varpi}_{k,n,f}}_{\text{Term 2}}}_{\text{Penalty Function}},
\end{align}
where ``Term 1" of the penalty function is to penalties of the relaxed  binary variables and ``Term 2"  is for rank-one solution discussed above.
}
%~The tightness of the SDP relaxation is verified by the following proposition.
%\begin{proposition}
%	The optimal beamforming matrix $\tilde{\bold{W}}$ is a rank-one matrix while $P_{\text{max} }\geq 0$.
%\end{proposition}
%\begin{proof}
%	Please see Appendix C.
%\end{proof}
\textcolor{blue}{
It is worth mentioning that the final optimization problem
%\eqref{RankF}
is an standard SDP programming which
can be solved optimally using CVX. Therefore, an iterative algorithm can be employed to tighten the obtained
upper bound of \eqref{RankF} in iteration $t$ is used as an approximation point for the next iteration $(t+1)$~\cite{MISO_Optimal,Ata_TWC}. 
The main steps of the proposed solution algorithm are listed  in Algorithm \ref{Dinkelbach algorithm}. Next, we discuss the initialization algorithm and optimally  of the proposed solution. 
\begin{algorithm}[t]
	\small \caption{\small Proposed Iterative Algorithm  }
	\label{Dinkelbach algorithm}
	\centering
	\begin{algorithmic}[1]
		\REQUIRE{$q_0=0,~t=0,~\epsilon>0$, \textcolor{blue}{initialize feasible points as \textit{described in Section \ref{Ini}  }}
		\STATE $\quad$ $q_t$ : Dinkelbach parameter
		\STATE $\quad$ $t$ : Iteration index 
		\STATE $\quad$ $\epsilon$ : The maximum tolerance
		\WHILE {$q_t - q_{t-1} > \epsilon$}
		\STATE Obtain resource allocation policy through solving problem (\ref{DC_Final}) 
		\STATE  Set $t=t+1$ 
		\STATE  Set {$q^{\ast}=\frac{\log_2(e)\sum_{f \in \mathcal{F}}\sum_{ k\in \mathcal{K}}\sum_{ n \in \mathcal{N}}({a_{f,k,n}^{*}}-b^{*}_{f,k,n})}{P_\text{Total}(\bold{W}^{\ast})}=\max_{\Xi'} \eqref{RankF}
	%	\bold{W'},\tilde{\bold{W}},\boldsymbol{\xi},\boldsymbol{\rho},\bold{x},\bold{a},\bold{b}}
%		\frac{R_{\text{Total}}(\bold{W'},\tilde{\bold{W}},\boldsymbol{\xi},\boldsymbol{\rho},\bold{x},\bold{a},\bold{b})}{P_\text{Total}(\bold{W'},\tilde{\bold{W}},\boldsym% bol{\xi},\boldsymbol{\rho},\bold{x},\bold{a},\bold{b})}
		$}
		\ENDWHILE
		\STATE Set \{$\bold{W'}^{\ast},\tilde{\bold{W}}^{\ast},\boldsymbol{\xi}^{\ast},\boldsymbol{\rho}^{\ast},\bold{x}^{\ast},\bold{b}^{\ast}$\} =$ \{\bold{W'}^{t-1},\tilde{\bold{W}}^{t-1},\boldsymbol{\xi}^{t-1},\boldsymbol{\rho}^{t-1},\bold{x}^{t-1},\bold{b}^{t-1}$\}
		}
		\RETURN
	\end{algorithmic}
\end{algorithm}
}
\textcolor{blue}{
\subsubsection{Initialization Algorithm}\label{Ini}
Due to the existence of Taylor approximations,  we should determine the initial feasible values for relevant variables.
The initial point for the relaxed variables denoted as
%binary  variables, i.e.,
$\boldsymbol{\rho}^0=[\rho_{k,n}^{f,0}],~ \boldsymbol{\xi}^0=[\xi_{k,n}^{f,0}],~ \bold{X}^0=[x_{i,n}^{k,f,0}]$, beam-forming variables, i.e., $\bold{W}^{0},~\bold{W'}^{0},~\text{and}~\tilde{\bold{W}}^{0}$ (subscripts are removed for simplicity), and  auxiliary variables, i.e., $\bold{b}^0=[b_{f,k,n}^{0}],~\bold{e}^{0}=[e_{i,n,f'}^{0}],~\text{and}~\bold{g}^{0}=[g_{i,k,n,f'}^0]$. 
We randomly generated $\boldsymbol{\rho}^0,~ \boldsymbol{\xi}^0,~\text{and}~ \bold{X}^0$ between zero and one. For beamforming variable, it should satisfy the power budget constraint. Therefore,  we set $\bold{w}^{0}=\sqrt{\frac{P_{\max}^1}{M}}\cdot \bold{r}_{M\times 1}$, where  $P_{\max}^1$ denotes the power budget of the macro AAU, $M$ is the number of antennas, and $\bold{r}_{M\times 1}$ is the vector with size $M\times 1$ and random elements between zero and one.
Now, according to the previous definitions, we can set  $\bold{W}_{k,n,f}^{0}=\bold{w}^{0}(\bold{w}^{0})^{\dagger}~$, $\bold{W'}_{i,k,n,f}^{0}={x}_{i,n}^{k,f,0}\bold{W}_{k,n,f}^{0},~\text{and}~\tilde{\bold{W}}_{k,n,f}^{0}=\rho_{k,n}^{f,0}\bold{W}_{k,n,f}^{0},~\forall f,k,i,n$.
The feasible points for the auxiliary variables can be set as follows:
\begin{align}
  &  b_{f,k,n}^{0}\triangleq \phi_{f,k,n}^0=\ln\Big(\sum_{f' \in \mathcal{F}}\sum_{i\in \mathcal{K}\backslash \{k\}}\text{Tr}[\tilde{\bold{H}}_{k,n,f'}\tilde{\bold{W}}_{i,n,f'}^{0}]\nonumber\\ &-\text{Tr}[\tilde{\bold{H}}_{k,n,f'}\bold{W'}_{i,k,n,f'}^{0}]\nonumber
\\&
+
 e_{k,n}^{f'}\Big(\mathcal{A}(\bold{W}_{i,n,f'}^{0})-\mathcal{B}(\bold{W^{'}}_{i,k,n,f'}^{0})\Big)+\sigma^2_{k,n,f},\Big).
\end{align}
Similarly, the feasible points of the auxiliary variables for handling of the SIC ordering constraint \eqref{SIC_C_u} are obtained by
%Note that for avoid of repetition and simplicity as described before, we do not bring the exact equations.
\begin{align}
   & e_{i,n,f'}^{0}\triangleq \ln(\mathcal{E}_{i,n,f'}^{0}),\forall i,n,f',
    \\
   & g_{i,k,n,f'}^0\triangleq\ln(\mathcal{G}_{i,k,n,f'}^{0}),
\end{align}
where $\mathcal{E}_{i,n,f'}^{0}$ and $\mathcal{G}_{i,k,n,f'}^{0}$ are obtained by \eqref{SICE} and \eqref{SICG}, respectively,  with replacing the above initial values. 
 However, the randomly generated values may not be a feasible solution. In such cases, we generate a new one, until we find a feasible point. 
It is worthwhile mentioning that this method may not be efficient, especially, when we have more constraints and high dimensional variables. For such cases, the initialization algorithm proposed in \cite{STAR} can be utilized. Further, sometimes  for a given resources, e.g., power budget, the problem may not be feasible. In this case, the elasticization method can be used \cite{ElasBook}.
}
\indent
\textcolor{blue}{
\subsubsection{Optimality Analysis}
In the following, we discuss about the optimally of the proposed algorithm.  As discussed before, we resort some approximations and relaxation methods,
i.e., MM, Taylor series expansion, and
SDR technique, to transform the original problem into a convex
problem. 
%Hence, the achieved solution is not be optimal in general. 
%However, the methods do not diminish the feasible set of the original problem too much. 
%Also, the SDR technique is tight \cite{DerrikIOT,STAR}. 
In particular, we apply a penalty factor for both relaxation and SDR to guarantee tightness of solution \cite{DerrikIOT}. More specifically, we consider two penalty functions
in the new objective function.~In fact, the penalty factors $ \lambda $ and $ \mu $ are
adopted to penalize the objective function when the integer values as well as rank-one solution are not available. Hence, for large values of the penalty parameters, the value of relaxation form of the binary variables are binary and beamforming matrix is rank-one. In this case, the maximum point of the new objective function %i.e., \eqref{RankF} 
is the maximum point of the original problem \cite{STAR,DerrikIOT}.
We further adopt the MM technique that may not approach the globally optimum solution. However,
the solution achieves a closely optimal solution due to the performance of the MM algorithm \cite{FNOMA}.
}
%we resort  some approximation and relaxations methods,
%i.e., MM,
%\footnote{In the literature, it has been proved that the MM approach
 %	achieves some close to optimal solution \cite{MISO_Optimal,KNOMA}.},
 %	Taylor series expansion and
%SDR technique, to transform the original problem into a convex
% problem. Hence, the achieved solution maybe not optimal. However, the methods do not diminish the feasible set of the original problem too much. Also, the SDR technique is tight. At the same time, we propose a penalty factor for both relaxation and SDR to guarantee that they be tight \cite{DerrikIOT}. In particular, we consider two penalty function in the new objective function.   In fact, the penalty factor is forced to be zero with large values for penalty parameters, i.e., ($ \lambda $ and $ \mu $). Hence, for large value of the penalty parameters, the value of relation variables are exact integer values and beamforming matrix be rank-one. In this case, the maximum point of the new objective function, %i.e., \eqref{RankF} 
% is the maximum point of the original problem. 
%On the other hand,
%for large 
%the penalty factors go to infinity, the minimum point of the new
%objective function is the optimal solution of the original problem.
%the solution in such way is near-optimal due to 
%the    successive convex approximation method is not far too much from the globally optimal solution \cite{MISO_Optimal,KNOMA,FNOMA}.
%}
 	\subsection{Low Complexity Algorithm Design}
% \textcolor{blue}{
 	It can be perceived that Algorithm 1 achieves a close to optimal 
 solution.~However, it may not suitable for large scale resource allocation with limited computational complexity.~Now, we aim at designing a low complexity algorithm for improving the practicality of Algorithm 1.~The proposed low-complexity algorithm is based on the heuristic solution known as two-step iterative approach.~In particular, the original problem is decomposed into two sub-problems, namely: 1) scheduling (i.e., joint user association and subcarrier allocation) and 2) beamforming design and SIC ordering.~Each subproblem can be solved while fixing the variables of other problems.~The  iterative procedure is adopted to find the scheduling, beamforming design, and SIC ordering. %~We exploit the update rule in \eqref{update_ASM} to find the resource allocation design. 
%}
%\textcolor{blue}{
% 	\begin{align}
% 	\nonumber
% 	%\label{update_ASM}
% &\underbrace{\boldsymbol{\rho}[0]\rightarrow\textbf{W}[0],\boldsymbol{\xi}[0]}_{\text{Initialization}} ...  \underbrace{\rightarrow\boldsymbol{\rho}[l-1]\rightarrow\textbf{w}[l-1],\boldsymbol{\xi}[l-1]}_{\text{Iteration \textit{l-}1}}\nonumber\\ &\rightarrow \underbrace{\boldsymbol{\rho}[l]\rightarrow\textbf{w}[l],\boldsymbol{\xi}}_{\text{Iteration \textit{l}}}\rightarrow ...\rightarrow\underbrace{\boldsymbol{\rho}^{\star}\rightarrow\textbf{w}^{\star},\boldsymbol{\xi}^{\star}}_{\text{Sub-Optimal Solution}},
% \end{align}
%}
\subsubsection{Solution of the Scheduling Subproblem}
\textcolor{blue}{
By assuming fixed beamforming design and SIC ordering parameters,
 the scheduling sub-problem is formulated as follows:
\begin{subequations}\label{subproblem_sic_order_main_new}
	\begin{align}\label{Reeq66}
	\max_{\boldsymbol{\rho}}\;&~\eta_{\text{EE}}^{\text{Worst}}
	\\\label{power_cons.1}
	\text{s.t.}~~&
	%	\begin{array}
	%\text{C}_{1}:\sum_{k\in \mathcal{K}}\sum_{n\in \mathcal{K}}\rho_{k,n}^{f}\|\bold{w}_{k,n,f}\|_{2}^{2}\le P_{\text{max}}^{f},\forall f\in\mathcal{F},
%	\\&\label{NOMA_Con1}
	\text{C}_{2}:\sum_{ k\in \mathcal{K}}\rho_{k,n}^{f}\le L_{n}^{f},\,\forall n\in\mathcal{N}, f\in\mathcal{F},
%	\\&
%	\label{BS_Ass_Con}	
%	\text{C}_{3}:\rho_{k,n}^{f}+\sum_{f'\in\mathcal{F}\backslash{\{f\}}}\rho_{k,n}^{f'}\le 1,\forall f\in \mathcal{F},n\in \mathcal{N},{k}\in\mathcal{K},
	%\\&\label{}
	%	\text{C}_{5}:\eqref{SIC_Ordering_Con}
	\\&\label{Ree9i1}
	\text{C}_{4}:
\rho_{k,n}^{f}\in
	\begin{Bmatrix}
	0,
	1
	\end{Bmatrix},\,\,\forall k\in \mathcal{K}, n\in\mathcal{N}, \forall f\in\mathcal{F},
	\\&\label{SICs.21}
	%\min_{|\mathbf{h}_{m,n}|\in \omega_{m,n}}
 \nonumber
	\eqref{User_Assosiaction_2}, \eqref{7final}, \eqref{SIC_Con_1}, \eqref{SIC_Con}.  % \eqref{SIC_Or_Num},  \eqref{Com_Delay}
	%	\end{array}
	%\\&\label{ratemin}r_m\ge r_{min},\,\,\,\forall m\in\mathcal{K}.
	\end{align}
\end{subequations}
}
\textcolor{blue}{
Problem \eqref{subproblem_sic_order_main_new} is integer non-linear programming. To solve it,
we propose a low-complex modified \textit{two-sided many-to-many} matching algorithm \cite{Zakeri}. As stated before, the scheduling variable determines both AAU selection and subcarrier assignment. To complete it, we propose two-stage matching algorithm\cite{Zakeri, Matching_2, Matching_II, Matching_1}. In the proposed matching algorithm, in the first stage,   we match users to the AAUs (``AAU selection'') and in the second stage, we match users of each AAUs to the subcarriers (``subcarrier assignment''). In doing so, each user $k\in\mathcal{K}$ constructs its  own preference list of AAUs denoted by $ \mathcal{L}_{k}^{\text{AAU}}=\{l_{k,f}\},\forall k,f, $ based on the path loss, i.e., the nearest AAU has the maximum preference and is the first in list $ \mathcal{L}_{k}^{\text{AAU}} $. 
%Based on this, users are paired with highly preferred AAUs until the power budget of it allows that be matched. 
All paired users of each AAU $ f $ is  indicated by $ \mathcal{K}_f $ which is the output of the first stage of matching process. After that each user in $\mathcal{K}_f   $ regenerates the preference list with respect to the each subcarrier $n$ based on the  $ |\bold{h}_{k,n,f}^{\dagger}\bold{w}_{k,n,f}|,~\bold{h}_{k,n,f}\in\mathcal{H}_{k,n,f},~ \forall k\in\mathcal{K}_f $ as a matching criteria. After that, the second stage of matching process is started to assign  subcarriers to users  in each AAU. 
}
%This matching process is performed  based on the \textcolor{blue}{ similar steps of [Algorithm 1, \cite{Zakeri}].}
\begin{proposition}
Each stage of the adopted matching algorithm after a few numbers of iterations will be a two-sided exchange-stable matching. Therefore, the devised two-stage matching algorithm is an stable matching algorithm.  
\end{proposition}
\begin{proof}
	Please see [Proposition 1, \cite{Zakeri}].
\end{proof}
\subsubsection{Beamforming Design and SIC Ordering Subproblems}	 
For the given scheduling, we aim to solve the problem of beamforming design and SIC ordering. To this end, we first define a 
%change the definition of $\tilde{\bold{W}}_{i,k,n,f}$ 
%	 	\begin{subequations}\label{Ssic_order_change}
%	\begin{align}
%		%\label{Reeq61}
%	\max_{\bold{W},\boldsymbol{\xi}}\;&~R_{\text{Total}}^{\text{Worst}}(\bold{W},\boldsymbol{\xi},\boldsymbol{\rho}^{l-1})-q~P_\text{Total}(\bold{W},\boldsymbol{\xi},\boldsymbol{\rho}^{l-1})
%	\\
%	\label{power_cons.111}
%	\text{s.t.}~
%	&\text{C}_{2}:\sum_{k\in \mathcal{K}}\sum_{n\in \mathcal{K}}\text{Tr}[\rho_{k,n}^{f}\bold{W}_{k,n,f}]\leq P^{f}_{\text{max} },\\
%	&\text{C}_{3}:\|\boldsymbol{\epsilon}_{k,n,f}\|_{2}^{2}\le \boldsymbol{\Delta}_{k,n,f},\\
%	&\text{C}_{4}:
%	\xi_{i,n}^{k} \in
%	\begin{Bmatrix}
%	0,
%	1
%	\end{Bmatrix}\\
%	%\\&\label{sic_con1}\text{C}_{5}:\sum_{k\in\mathcal{K}}x_{k,n}^{f}\leq 1,\,\forall n\in\mathcal{N},\\
%	&\text{C}_{5}:\bold{W}_{k,n,f}\succeq \bold{0},\\
%	& \text{C}_{6}:\text{Rank}(\bold{W}_{k,n,f})\leq 1,\\
%	%+\lambda \Bigg(\sum_{f\in \mathcal{F}}\sum_{i,k\in \mathcal{K}}\sum_{ n \in \mathcal{N}} \xi_{i,n}^{k}-(\xi_{i,n}^{k})^{2} +\sum_{f\in \mathcal{F}}\sum_{ k\in \mathcal{K}}\sum_{ n \in \mathcal{N}} \rho_{k,n}^{f}\nonumber\\&-(\rho_{k,n}^{f})^{2}+\sum_{f\in \mathcal{F}}\sum_{ k\in \mathcal{K}}\sum_{ n \in \mathcal{N}} x_{k,n}^{f}-(x_{k,n}^{f})^{2}\Bigg)\\
%	&\eqref{SIC_Sub_Con}, \eqref{SIC_Con_1}, \eqref{Fronthaul_Con}, \eqref{SIC_C}.	
%	\end{align}
%\end{subequations}
 new varaible  as $\tilde{\bold{W}}_{i,k,n,f}\triangleq \xi^{k}_{i,n}\bold{W}_{k,n,f}$. In a similar manner, we adopt the same approach for obtaining the beamforming design as well as SIC ordering.~Consequently, the problem at hand can be written mathematically as 
 \begin{subequations}\label{Change_DC_Problem}
	\begin{align}
	\nonumber
\max_{\tilde{\bold{W}},{\bold{W}},\boldsymbol{\xi},\bold{a},\bold{b}}&\sum_{f \in \mathcal{F}}\sum_{ k\in \mathcal{K}}\sum_{ n \in \mathcal{N}}({a_{f,k,n}}-b_{f,k,n})\log_2(e)-q\cdot P_\text{Total}(\bold{W})\\&-f(\boldsymbol{\xi},\lambda)+\tilde{g}(\boldsymbol{\xi},\lambda)-{\mu}\times \sum_{f \in \mathcal{F}}\sum_{ k\in \mathcal{K}}\sum_{ n \in \mathcal{N}}  \tilde{\varpi}_{k,n,f}
	\\\label{PC}
	\text{s.t.}~~
	&\text{C}_{1}:\sum_{k\in \mathcal{K}}\sum_{n\in \mathcal{K}}\text{Tr}[\rho_{k,n}^{f}{\bold{W}}_{k,n,f}]\leq P^{f}_{\text{max}},~ \forall f \in \mathcal{F},\\	&\text{C}_{2}:\bold{W}_{k,n,f}\succeq \bold{0},\\
	& \text{C}_{3}: \exp ({a_{f,k,n}})\geq \sigma^{2}_{k,n,f},~\exp({b_{f,k,n}})\geq\sigma^{2}_{k,n,f}\\
	& \text{C}_{4}: a_{f,k,n}\geq b_{f,k,n}\label{a_bound}\\
	& \text{C}_{5}: \psi_{f,k,n}\geq \exp(a_{f,k,n}),~\phi_{f,k,n}\leq \exp(b_{f,k,n}),\\
	& \text{C}_{6}:~\tilde{\bold{W}}_{i,k,n,f}\preceq  P_{\text{max}}^{f} \bold{I}_{M}\xi_{i,n}^{k},\\
	&\text{C}_{7}:~\tilde{\bold{W}}_{i,k,n,f}\preceq \bold{W}_{k,n,f},\\
	&\text{C}_{8}:~\tilde{\bold{W}}_{i,k,n,f}\succeq \bold{W}_{k,n,f}-(1-\xi_{i,n}^{k}) P_{\text{max} } \bold{I}_{M},\\
	&\text{C}_{9}:~\tilde{\bold{W}}_{i,k,n,f}\succeq \bold{0},\\
	& \text{C}_{10}:~0\leq\xi_{i,n}^{k}\leq 1,\\
	& \eqref{SIC_Con_1}, \eqref{Fronthualu1}, \textcolor{blue}{\eqref{SIC_Con}}, \eqref{7final},
	\end{align}
\end{subequations}
where $\tilde{\bold{W}}$ stands for all of $\tilde{\bold{W}}_{i,k,n,f},~\forall k,i,n,f,~i\neq k$.
This problem is convex and can be solved via efficient convex programming libraries like CVX.
\textcolor{blue}{
\subsection{Complexity Analysis of the Solution Algorithms}
This section provides the complexity analysis and the  comparison of the proposed solution algorithms. 
In the first algorithm, i.e., Algorithm \ref{Dinkelbach algorithm}, we solve the original problem in one step as the form of \eqref{DC_Final} via CVX. In this problem, there are totally $7KNF+5K^2NF$ variables and $10KNF+8K^2NF+4K^2NFM^2+F+2KN+4KNFM^2$ convex and affine constraints. 
Note that the term $K^2$ instead of $K$ is from considering the SIC ordering variable and the resulting constraints.
%According to IPM method for solving SDP problem in CVX, the computational complexity for $\alpha-$optimal soltion \cite{DerrikIOT}
Thus, the complexity of the algorithm per iteration is $\mathcal{O}\big((7KNF+5K^2NF)^3(10KNF+8K^2NF+4K^2NFM^2+F+2KN+4KNFM^2)\big)$ \cite{Comp1}. Therefore, by considering $K>F~\text{and}~N> F$, which is logical for practical setting, and for sufficiently large values of $K$ and $N$, the overall order of the complexity of our first algorithm can be calculated by $\mathcal{O}\big((K^2NF)^4\big)$.
%\\\indent
In the second  algorithm, we applied the  alternating approach. Based on this, the overall complexity is a linear combination of the complexity of solution of
each sub-problem. The solution of the first sub-problem is a matching algorithm whose complexity  is a linear function of the number of the sub-carriers, users,  and AAUs, i.e., $\mathcal{O}(K\times N\times F)$. For the second sub-problem \eqref{Change_DC_Problem}, we also applied  a similar approach as the first algorithm, but the number and dimension of variables as well as constraints are considerably reduced. Problem \eqref{Change_DC_Problem} includes
 $4KNF+3K^2NF$ variables and $7KNF+5K^2NFM^2+2F$ convex and affine constraints. Based on the solution algorithm of \eqref{Change_DC_Problem}, the computational complexity per iteration is $\mathcal{O}\big((4KNF+3K^2NF)^2(7KNF+5K^2NFM^2+2F)\big)$. Thus, the overall complexity is $\mathcal{O}\big((K^2NF)^3\big)$.
As a result, both of the proposed algorithms have a polynomial order of complexity, whereas the overall complexity of an exhaustive method is exponential over the number of constraints and search variables.
%Furthermore, based on interior method steps which is used by CVX
}
 %In the flowchart below, the steps of our algorithm are illustrated.
 %%%%%%%%%%%%%%%%%%%%%%%%%%%%%%%%%%%%%%%%%%%%
 \tikzstyle{decision} = [ellipse, draw, fill=pink!18,
 text width=15em, text centered, rounded corners, minimum height=3.5em]
 \tikzstyle{block} = [rectangle, draw, fill=blue!10,
 text width=25em, text centered, rounded corners, minimum height=4.5em]
 \tikzstyle{line} = [draw, -latex']
 \tikzstyle{cloud} = [draw, ellipse,fill=red!20, node distance=2cm,
 minimum height=2em]
% \begin{center}
% 	\begin{tikzpicture}[node distance = 2.3cm, auto]
% 	\node [decision] (init) {Original optimization problem (20)};
% 	\node [block, below of=init] (c1) {\textbf{Step 1: } Obtain worst-case data rate based on the Proposition 1 };
% 	\node [block, below of=c1] (xs) {\textbf{Step 2: }Define new variables as $x^{k,f}_{i,n}$ and adopt linearization technique based on the equatition (29)-(31)};
% 	\node [block, below of=xs] (continuous) {\textbf{Step 3: }Convert integer variables into continuous variables using (32)-(34)};
% 	\node [block, below of=continuous] (objective) {\textbf{Step 4: }Employ big-M method (41)-(48) };
% 	\node [block, below of=objective] (rate) {\textbf{Step 5: }Find a convex approximation based on MM approach using (54)-(55)};
% 	\node [decision, below of=rate] (last) {Solve Convex optimization problem (56)};
% 	% Draw edges
% 	\path [line] (init) -- (c1);
% 	\path [line] (c1) -- (xs);
% 	\path [line] (xs) -- (continuous);
% 	\path [line] (continuous) -- (objective);
% 	\path [line] (objective) -- (rate);
% 	\path [line] (rate) -- (last);
% 	\end{tikzpicture}
% \end{center}
% ===
%===================
%\\
%\textcolor{blue}{I want to approximate the matrix norm 2 function via first Taylor Approximation:
%$  F(\mathbf{A})=\|\mathbf{A}\|  $
%\begin{align}
% F(\mathbf{A})= \|\mathbf{A}_0\|+\text{trace} (2\mathbf{A}_0.*(\mathbf{A}-\mathbf{A}_0)),
%\end{align}
%where $  \mathbf{A}_0$ is the initialization point.
%}
%\vspace{-1.8em}
\section{Numerical Evaluation}\label{Simulation_results}
%\textcolor{blue}{
	This section presents numerical results to assess and compare the designed SIC ordering and beamforming scheme under various configurations which makes comparisons with conventional ones. We provide numerical results regarding to different metrics such as energy efficiency and utilized power under variation of different parameters. 
\subsection{Simulation Setup}
We consider a C-RAN network such that a high power AAU is located at the center of a service coverage area  with $ 500 $ m radius, and $ 3 $ 
low power AAUs  with a circular coverage area with $ 20 $ m radius which are randomly located \cite{Ro_Moltafet}.  Also, the number of total users is $K=8$  and the number of antennas for each AAU is $M=3$ \cite{Ro3,Good_Robust}. 
% Channel coefficients between BS and user $k$ is i.i.d Gaussian distribution as  $\bold{h}_{k,n,f}\sim \mathcal{CN}(\bold{0},\bold{I})$ \cite{bibid}.
Unless otherwise stated, the simulation results are based on values of the parameters which are listed in Table \ref{Simulation_Sett}. Moreover, the small-scale fading of the channels is assumed to be Rayleigh fading and the large-scale fading effect is denoted by $d_{k,f}^{-\alpha}$ to incorporate the path-loss effects, where $d_{k,f}$ is the distance between user $k$ and AAU $f$ measured in meters, and $\alpha=3$ is the path-loss exponent \cite{Ro3}.
	\begin{table}
		\tiny
%	\small
	\centering
	\caption{Main parameters values for network setup }%\cite{Ro3,Robust,Good_Robust}}
	\label{Simulation_Sett}
	\begin{tabular}{ c|c}
		\hline
		\textbf{Parameter(s)}  & \textbf{Value(s)}  \\
		\hline
		$K/N/M/F$  &$8/5/3/4$
		\\
		\hline
		\text{Coverage radius of Macro AAU/Small AAU} &   $500/20$\,\text{m}
				\\
		\hline
	Power budget of Macro AAU/Small AAU/$C_{\text{Circuit}} $&   $40/30/30$\,\text{dBm}
		\\
		\hline
		$\lambda/\beta  $& %$10^{(M+\log(\frac{P_{\max}^{f}}{\sigma_{f,k,n}^2}))}
		$10^{5}$/$\frac{1}{4}$
		%-----
		\\\hline
		$\sigma_{f,k,n}^2/\delta_{k,n}^f$& 
		$-174~ \text{dBm/Hz}/0.001$
		\\
		\hline
	\end{tabular}
\end{table}
%\vspace{-1.8em}
\subsection{Results Discussions}
In this subsection, we discuss about the simulation results achieved for the following main scenarios:
\begin{enumerate}
	\item \textit{Proposed near-optimal solution and SIC ordering method (\textbf{Proposed algorithm})}
	\item \textit{Proposed two-step solution and SIC ordering as a baseline (\textbf{Baseline 1}):} In this baseline, the main problem is solved iteratively in which the scheduling variable is obtained by the devised   matching algorithm. 
	\item\textit{SIC ordering based on channel gains as a baseline (\textbf{Baseline 2}):} In this baseline, the decoding of users is based on the absolute value of channel gains as considered in \cite{Ro2,Dai,FNOMA,MISO_Optimal,Moltafet,Minorization,ICC,Sharma}. Note that in this algorithm, SIC constraints lead to the lower achievable rates. 
\end{enumerate}
%% Convergence Analysis   START
%	 	 \begin{figure*}[t]
%	\centering
%	\includegraphics[width=.5\textwidth]{./Figures/SCA_Convergence}
%	\caption{Convergence of SCA (Algorithm 1)}
%	\label{SCA_Con}
%\end{figure*}
%	 	 \begin{figure*}[t]
%	\centering
%	\includegraphics[width=.5\textwidth]{./Figures/ASM_Convergence}
%	\caption{Convergence of iterative (Algorithm 2)}
%	\label{ASm_Con}
%\end{figure*}
%%% END
	 \begin{figure}[t]
	\centering
	\includegraphics[width=.45\textwidth]{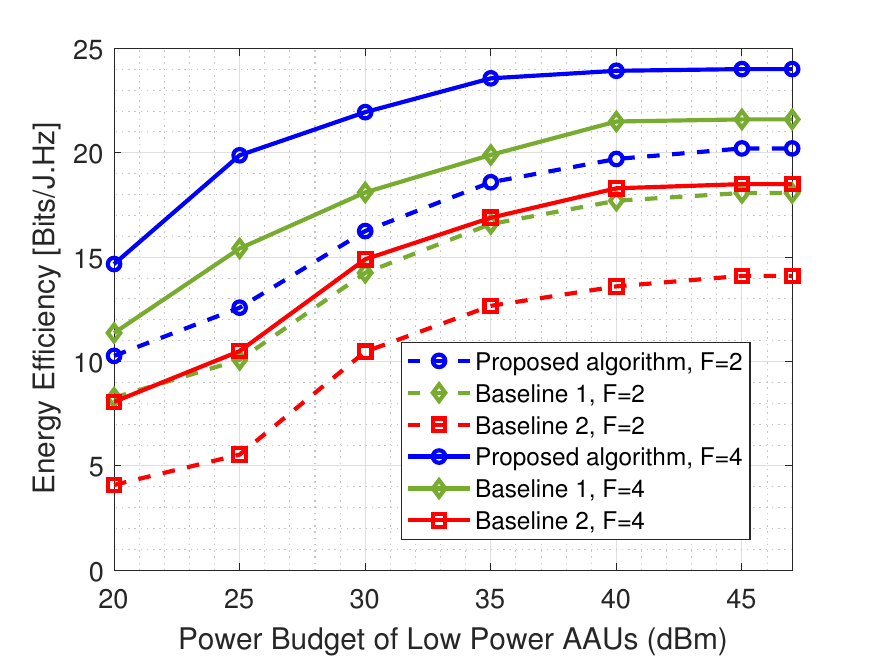}
	\caption{EE versus the value of power budget for each AAU for different number of AAUs.}
	\label{EE_Power}
\end{figure}
%\vspace{-1em}
All the above scenarios are investigated under different system parameters which are discussed in the following.  First, we investigate the effect of variation of power budget on the EE of the network for different number of AAUs in Fig.  \ref{EE_Power}.  In this figure, we change the power budget values from $ 20 $ dBm to  $ 47 $ dBm and also we observe that the EE first increases and then is saturated when the transmit power is larger than $35$ dBm, i.e., $ P_{\max}^f=35 $ dBm. This is because of exploiting power control via designing beamforming for all schemes, and also we can deduct that the beamforming works well to improve EE up to the maximum point. Besides, we observe that the performance of our proposed SIC ordering and beamforming algorithm in terms of EE significantly outperforms the other baselines. The main reason behind this achievement is that  the proposed   SIC ordering algorithm is performed via optimizing the SIC ordering variable which is exploited to handle the intra-cell and inter-cell interference to maximize EE. We also observe that our proposed algorithm has a better performance as compared to baseline 1 due to performing resource allocation design and SIC ordering jointly in a one-step optimization problem. While, in baseline 2,  SIC ordering is based on the absolute value of the channel gains. However, SIC ordering in baseline 2 is applicable with an acceptable performance guarantee for single antenna systems and cannot be applied for MISO NOMA systems, efficiently. Also, this figure investigates the effect of AAUs on the EE of the system.~It is evident that our proposed algorithm outperforms  other baseline schemes due to performing joint user association and subcarrier allocation which improves significantly the performance  of the system. 
%\\$\bullet$~\textit{Effect of  power budget:} 	
%\\$\bullet$~\textit{Effect of number of subcarriers:}
	 \begin{figure}[t]
	\centering
	\includegraphics[width=.45\textwidth]{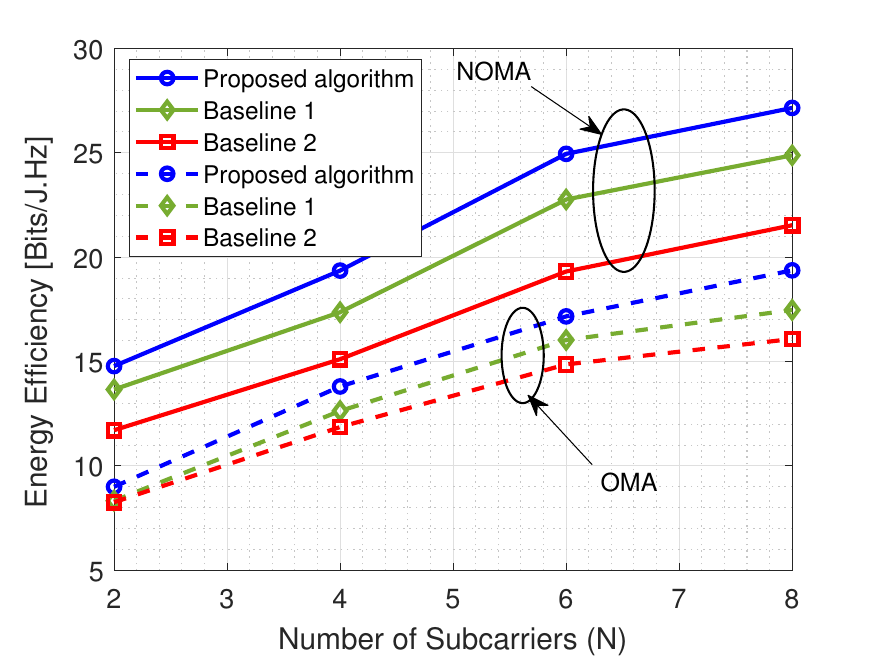}
	\caption{EE versus the number of shared  subcarriers  in each AAU for NOMA and OMA.}
	\label{EE_Subcarrier}
\end{figure}

In Fig. \ref{EE_Subcarrier}, we study the effect of the number of subcarriers and also the performance of NOMA as compared to the  conventional OMA on the baseline schemes. 
It is observed that the improvement of the proposed algorithm compared to the baselines is sustainable. This improvement is achieved   not only in the NOMA-based systems but also in the OMA-based systems. Note that in OMA, there is no intra-cell inference due to orthogonality in the utilization of the subcarriers. Besides,  our proposed SIC ordering controls the inter-cell interference.  Thus, our designed SIC ordering and beamforming are applicable not only for NOMA but also for any co-channel interference suffered communication networks without any need on the CSI of these channels. Consequently, the improvement of EE in NOMA is more than OMA. Also we can conclude that the performance of NOMA is much better than that of OMA due to exploiting each subcarrier more than one in the network. 
	  \begin{figure}[t]
	 	\centering
	 	\includegraphics[width=.45\textwidth]{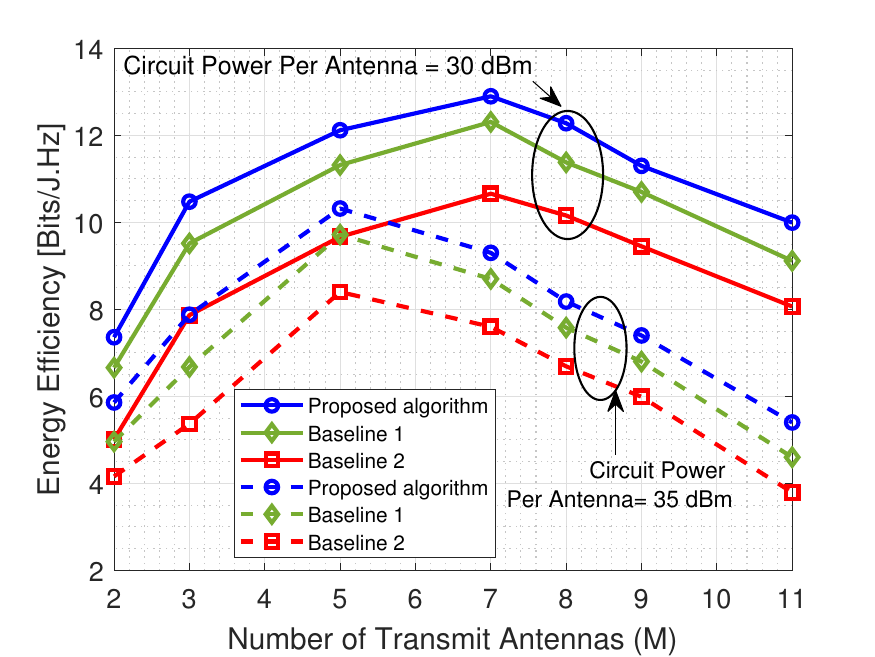}
	 	\caption{Energy efficiency versus the of number of antennas in each AAU.}
	 	\label{EE_Noantennas}
	 %	\vspace{-2em}
	 \end{figure}
In Fig. \ref{EE_Noantennas}, we evaluate the EE of the considered schemes while considering the effect of the number of antennas in each AAU for different circuit power values.
%
% the effect of the number of antennas in each AAU on the EE in considered scenarios for different circuit power values.  
% 
 As can be seen from this figure, EE increases as the number of antennas increases due to the array giants and spatial diversity, and then drops  after certain value for the number of antennas, i.e., $ M>7 $. 
%Also, note that the number is different for the values of circuit power and also other system parameters.  
This is because that employing more antenna 
enables high degrees of freedom in the spatial diversity gain which turns on improving SE while exceeding 
power consumption, specifically hardware power consumption which linearly increases as $ M $ increases  due to activating RF chain per each antenna. While the SE is changed slowly  with respect to the value of $ M $ which turns a strike a balance between SE and power consumption which leads to trade-off  between SE and EE. 
Further, for the higher values of circuit power, the maximum value of EE is obtained for a lower number of $ M $. This is because  that system’s aggregated power consumption has a
greater impact on improving system’s EE than maximizing the SE as SE commensurates to log-function.
% this is because of the contrariwise effect of hardware power in EE. Also, the maximum performance gap between the proposed algorithm and baselines is for the number of antennas such that EE is maximum. 
 The interesting results from this evaluation can be explained as two folds:
 First,  we need  an appropriate beamforming design  for massive antennas communication systems and second there is a need for designing an efficient antenna selection algorithm  to select appropriate antennas and then 
% 
% 
% 
%  we should design not only better beamforming for massive antennas communication systems but also better antennas selection algorithm to select appropriate antennas and then 
  doing precoding, especially for massive MIMO mmWave 6G networks. Also, this figure reveals  the performance of our SIC ordering algorithm for massive antennas networks. In our future work, we will propose  an appropriate antenna selection (finding optimum $M $) as well as beamforming design for massive mmWave networks.
	 \begin{figure}[t]
	\centering
	\includegraphics[width=.45\textwidth]{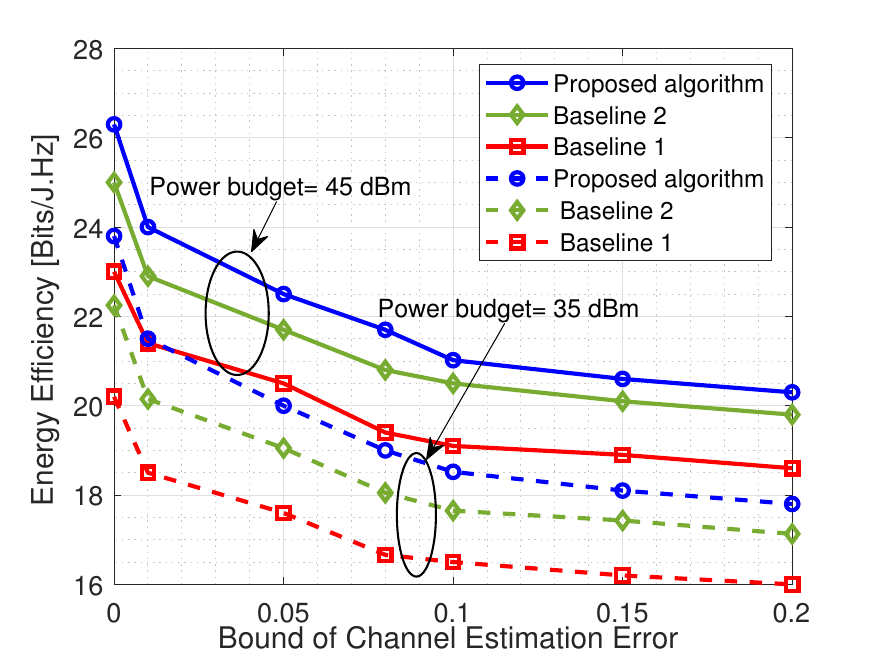}
	\caption{Impact of channel estimation error bound   on the EE.}
	\label{EE_Error_Power}
\end{figure}

Fig. \ref{EE_Error_Power} illustrates the impact of channel estimation error on the EE for different values of the power budget. 
% of high power AAU for the mentioned scenarios
 %where zero error is equivalent to the perfect CSI in which the complete information of channels of users are available in the BBU side. s
 %
As can be seen, the relation between the error bound and the system EE is indirect. Also, the upper bound is obtained  for  the perfect CSI setting\footnote{It is worthwhile to note that the zero error is equivalent to the perfect CSI in which the complete information of channels of users is available in the BBU side.}.
	It is seen that with increasing the error bound, our algorithm has a vigorous capability for deducting the impact of the imperfect  CSI. It can be also observe that the performance gap with respect to the baseline schemes becomes significantly large. This is because  of the existing indirect efficacy of the error bound on the performance of the SIC ordering based on the channel gains. It is worthwhile to note that performing SIC based on the channel gains needs  full CSI which is not practical in the real wireless communication networks.
%\textcolor{red}{When the error bound values have involved, the gap between the proposed algorithm and baselines is increased. This observation is the result of the direct effect of channel error on being of the SIC ordering based on the channel gain imperfect.} 
Furthermore, more power consumption is needed for large  values of the error bound to reach the same SE which makes the reduction on EE with increasing  the error bound. In other words, for a fixed value of consumed power, the SE tends to  low values for the higher   error bounds which results in the  EE reduction.
%}
	 \begin{figure}[t]
	\centering
	\includegraphics[width=.45\textwidth]{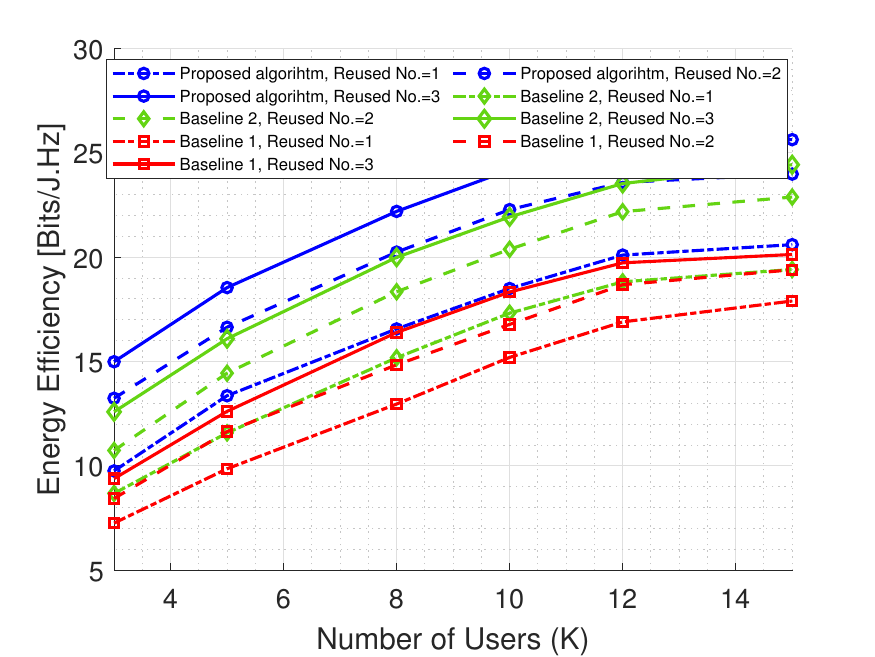}
	\caption{EE versus the number of users for different reused numbers.}
		%	Effect of number of users and maximum number of  reuse of each subcarrier in EE.}}
	\label{EE_User_Reused}
	%\vspace{-2em}
\end{figure}
Moreover, we study the behavior of EE achieved via the scenarios with respect to the number of users and maximum reuse number of each subcarrier in each AAU which is plotted in Fig.  \ref{EE_User_Reused}. 
Note that when the reuse number is $ 1 $, the considered network operates as OMA while for values of $ 2 $ and $ 3 $, the network operates as NOMA.  From this figure, it can be observed  that the EE grows with  the number of users and    reuse factor of NOMA because of multi user diversity. 
%and efficiently spectrum utilization. 
In addition, for the higher number of reuse factors in NOMA, the improvement on the EE is low. % as compared to the small values. 
% for the number of reused is $ 3 $ the improvement on EE  is low compared to reused $ 2 $. 
 The reason behind this trend is that the high reuse factors in NOMA
 boosts the denominator of the system throughput
 due to incorporating intra-cell interference in the data rate which results in an exceeding power consumption and consequently degrading the EE of the system.
 %
% the effect of intra-cell interference which degrades the performance gain due to consuming more power consumption which results in a decrement EE of the system.
  Furthermore, we can declare that for the high reuse factor, it is better to adopt clustering, i.e., user pairing methods, especially for the large number of users (e.g., massive connections). 
	 
\textcolor{blue}{Finally, we investigate the performance gap between the exactly robust solution denoted by ExRS (approach in Remark \ref{EXRS}) and strictly bounded robust solution denoted by SBRS, and the behavior of the introduced penalty function in \eqref{RankF}. For the First, Fig. \ref{Ex_SB_Com} shows the performance comparison between ExRS and SBRS under variation of channel estimation error bound for different power budget for macro AAU denoted by $P_{\max}^1$. %This figure is obtained for the common setting of Table \ref{} in the manuscript. 
As can be seen from the figure, for low power budget and small values of the error approximately the performance of ExRS and SBRS are close to each other.
	 \begin{figure}[h!]
	\centering
	\includegraphics[width=.45\textwidth]{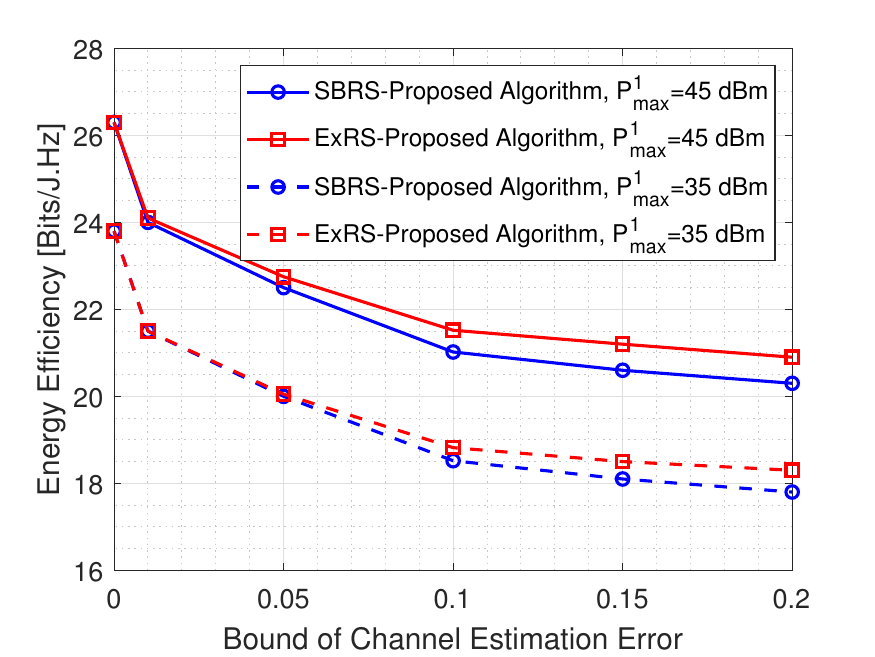}
	\caption{EE versus the channel estimation error bound  for different robust algorithms.}
	\label{Ex_SB_Com}
\end{figure}
%jhggggggggggggggggggggg
Moreover, the effect of penalty factors on the penalty function is examined in Fig. \ref{P_F_1}. In this figure, $P_{\max}^{1}=35~ \text{dBm},~K=8,~N=5$, $\lambda=\mu$, and both axes are plotted in the Logarithmic scale.
Note that in ``Without Round', the penalty function values is the penalty function in \eqref{RankF} (``Term 1+Term 2"), and
in ``With Round'',  the penalty function is only the penalty term for the rank-one (``Term 2" in \eqref{RankF})  due to $h(\boldsymbol{\xi},\boldsymbol{\rho},\bold{x},\lambda)=0$. As can be seen form Fig. \ref{P_F_1}, the penalty function values close to $0$, for the sufficiently large value of penalty factors, e.g., $10^{5}$, which ensures achieving a rank-one solution and the relaxed binary variables converges to binary ones. 
 	 \begin{figure}[h!]
	\centering
	\includegraphics[width=.45\textwidth]{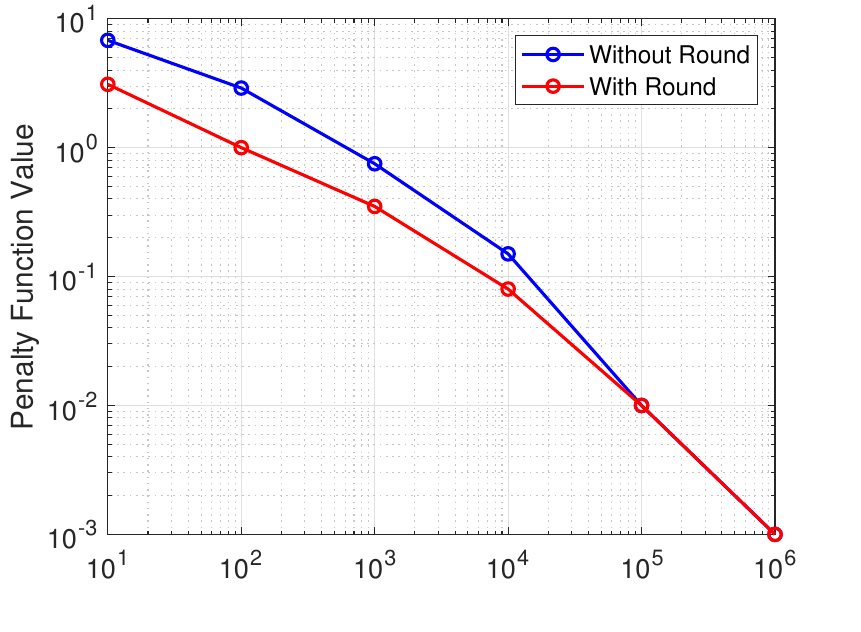}
	\caption{Impact of penalty factors on the penalty function of \eqref{RankF}.}
	\label{P_F_1}
\end{figure}
}

%$\bullet$~\textit{Effect of error of channel uncertainty:}
%$\bullet~$\textit{Power-budget of low-power AUUs}
%\subsection{Baseline Solution Algorithms} 
%\textbf{Base line:}  In first, each user should be assigned to at most one AAU which is imposed on the following constraint:
%\begin{align}
%& \rho_{k,n}^{f}+\sum_{f'\in\mathcal{F}\backslash{\{f\}}}\rho_{k,n'}^{f'}\le 1,\forall f\in \mathcal{F},n,n'\in \mathcal{N},{k}\in\mathcal{K}.
%\end{align} 
% % Alternative algorithm for Multi Connectivity
%\begin{align} 
%& \rho_{k,n}^{f}+\sum_{f'\in\mathcal{F}\backslash{\{f\}}}\rho_{k,n'}^{f'}\le 1,\forall f\in \mathcal{F},n,n'\in \mathcal{N},{k}\in\mathcal{K}.
%\end{align} 
% -------------------
%---------------------------
%The baseline system for comparison with NOMA could be
%an OMA setup when there are K = L users. This is similar
%to 100 of size K × K, and the UEs are scheduled over orthogonal
%resources.
%--------------------
%===================
%===================
 %\input{2020-05-02-Conclusion}
% \vspace{-1.8em}
 \section{Conclusion Remarks}\label{Conclusion}
 In this paper, we proposed a novel SIC ordering and also provided a robust and efficient algorithm for resource allocation and beamforming design for C-RAN assisted MC NOMA networks with imperfect CSI. In particular, we  formulated the worst-case EE by optimizing the SIC ordering, beamforming, and scheduling variables.
% under worst-case SIC ordering, proposed multi-connectivity, and power budget constraints. 
  Although, the  underlying  optimization problem is non-convex which is in the form of MINLP, we adopted  majarization-minimization and penalty factor methods to convert it into the convex one.  Furthermore, we provided a low complexity algorithm based on two-step iterative solution to  strike the balance between the complexity and performance gain. 
   Extensive simulations were provided to assess the performance of the  proposed algorithms. Moreover, simulation results unveil the superiority of the proposed algorithm as compared to the baseline schemes. 
  %Also, we investigate the impact of different network parameters such  as number of employed antennas. 
  %Simulation results demonstrate that the proposed robust EE algorithm can alleviate negative effect of imperfect  CSI, SIC,  limited power, and spectrum resources on the network performance.
  \\\indent
 The SIC ordering algorithm in NOMA-based communications has not been well addressed, especially restraining the inter-cell interference, and it would be a critical issue and pivotal impact on the performance of co-channel interference communication networks. 
 In order to broaden a new horizon that is inferred from the results for the future of massive antennas and high
 energy efficiency that necessitates the ubiquitous 6G, it is crucial to design efficient antenna selection, clustered
 beamforming, channel estimation, and spectrum management algorithms in a future wireless networks. 
% Further, in order to give future directions to the main point of achieved and discussed results for future the massive antennas and high energy efficient needed 6G networks, its crucial to design efficient antenna selection, clustered beamforming, channel estimation, and spectrum management algorithms in a  common network.
%\vspace{-1.2em}
 %===================
\section{Appendix}
\vspace{-3em}
\textcolor{blue}{\subsection{Proof of Proposition 1}
The proof includes two parts: 1) minimization and 2) maximization which are discussed as follows. 
\\\indent 
\textit{Proof of minimization:}
Using an arbitrary positive multiplier $\phi\ge 0$, the Lagrangian function of \eqref{minimumterm} can be written as
\begin{align}\nonumber
&\mathcal{L}(\boldsymbol{\Delta}_{k,n,f},\phi)=
\text{Tr}[(\tilde{\bold{H}}_{k,n,f}+\boldsymbol{\Delta}_{k,n,f})\bold{W}_{k,n,f}]\nonumber\\&+\phi(\text{Tr}~[\boldsymbol{\Delta}_{k,n,f}\boldsymbol{\Delta}_{k,n,f}^{\dagger}]-e_{k,n,f}).
\end{align} 
Setting the derivative of the Lagrangian with respect to $\boldsymbol{\Delta}_{k,n,f}$  to zero, we have:
\begin{equation}
\nabla_{\boldsymbol{\Delta}^{*}_{k,n,f}}\mathcal{L}(\boldsymbol{\Delta}_{k,n,f},\phi)=\bold{W}^{\dagger}_{k,n,f}+\phi\boldsymbol{\Delta}_{k,n,f}=0.
\end{equation}
The optimal value of $\boldsymbol{\Delta}_{k,n,f}$ is denoted by $\boldsymbol{\Delta}^{*}_{k,n,f}$ which can be obtained as
%\begin{equation}
$\boldsymbol{\Delta}^{*}_{k,n,f}=-\frac{1}{\phi}\bold{W}^{\dagger}_{k,n,f}.$
%\end{equation}
We also differentiate the Lagrangian function with respect to $\phi$ and equates it to zero as
%\begin{equation}
$\nabla_{\phi}\mathcal{L}(\boldsymbol{\Delta}_{k,n,f},\phi)=0,$
%\end{equation}
where the optimal solution for $\phi$ is given by
%\begin{equation}
$\phi^{*}=\frac{1}{e_{k,n,f}}\|\bold{W}^{\dagger}_{k,n,f}\|.$
%\end{equation}
By substituting the optimal value of $\phi$, i.e., $\phi^{*}$, we conclude that
\begin{equation}
\boldsymbol{\Delta}^{f,\text{min}}_{k,n}=-e_{k,n,f}\frac{\bold{W}^{\dagger}_{k,n,f}}{\|\bold{W}^{\dagger}_{k,n,f}\|}.
\end{equation}  
Hessian of the Lagrangian function verifies the obtained solution is minimum. To this end, we need to check the second derivative at the optimal point that should be positive semi-definite, i.e.,  \cite{Der_S}
\begin{align}
 &   \nabla^2_{\boldsymbol{\Delta}_{k,n,f}}\mathcal{L}(\boldsymbol{\Delta}_{k,n,f}^{*},\phi^{*})
    =\\&\frac{\|\bold{W}^{\dagger}_{k,n,f}\|}{e_{k,n,f}}\bigg(\text{vec}\{\bold{I}_{M}\}\text{vec}\{\bold{I}_{M}\}\bigg)^{T}\succeq \bold{0}.
\end{align}
\\\indent \textit{Proof of maximization:}
The Lagrangian of \eqref{21max} is given by
\begin{align}\nonumber
&\mathcal{L}(\boldsymbol{\Delta}_{k,n,f'},\phi)=
\text{Tr}[(\tilde{\bold{H}}_{k,n,f'}+\boldsymbol{\Delta}_{k,n,f'})\bold{W}_{i,n,f'}]\nonumber\\&-\phi(\text{Tr}~[\boldsymbol{\Delta}_{k,n,f'}\boldsymbol{\Delta}_{k,n,f'}^{\dagger}]-e^{2}_{k,n,f'}).
\end{align}
% Lagrangian function of \eqref{21max}
By differentiating above function with respect to $\boldsymbol{\Delta}_{k,n,f'}$ and setting the derivative to zero, we have:
\begin{equation}
\nabla_{\boldsymbol{\Delta}^{*}_{k,n,f'}}\mathcal{L}(\boldsymbol{\Delta}_{k,n,f'},\phi)=\bold{W}^{\dagger}_{i,n,f'}-\phi\boldsymbol{\Delta}_{k,n,f'}=0.
\end{equation}
We will found that $\boldsymbol{\Delta}^{*}_{k,n,f'}=\frac{1}{\phi}\bold{W}^{\dagger}_{i,n,f'}.$ Following the same steps for eliminating the role of $\phi$, we obtain that 
\begin{equation}
\boldsymbol{\Delta}^{f',\text{max}}_{k,n}=e_{k,n,f'}\frac{\bold{W}^{\dagger}_{i,n,f'}}{\|\bold{W}^{\dagger}_{i,n,f'}\|}.
\end{equation}
%which completes our proof.
Hessian of the Lagrangian function verifies that the obtained solution is maximum. Hence, we check the second derivative at the optimal point that should be negative semi-definite, i.e., \cite{Der_S} 
\begin{align}
\nonumber
   & \nabla^2_{\boldsymbol{\Delta}_{k,n,f}}\mathcal{L}(\boldsymbol{\Delta}_{k,n,f}^{*},\phi^{*})=
    \\&-\frac{\|\bold{W}^{\dagger}_{k,n,f}\|}{e_{k,n,f}}\bigg(\text{vec}\{\bold{I}_{M}\}\text{vec}\{\bold{I}_{M}\}\bigg)^{T}\preceq \bold{0}.
\end{align}
}
% --------------------
% \textcolor{blue}{
% \subsection{Proof of proposition 2}
% As see form \eqref{SIC_C} can be written as following from. 
\vspace{-1em}
% }
%============================
\textcolor{blue}{
\subsection{Proof of Proposition \ref{Pro_Binary}}
	We aim to prove this proposition by using the \textit{abstract Lagrangian} duality. The primal problem of (\ref{sic_order_change}) can be written as
%	\begin{equation}\label{61}
$	p^{*}=\max_{\bold{W},\boldsymbol{\xi},\boldsymbol{\rho},\bold{x}}\min_{\lambda}\mathcal{L}(\bold{W},\boldsymbol{\xi},\boldsymbol{\rho},\bold{x}),$
%	\end{equation}
	where the dual problem of (\ref{Reeq28}) is given by:
	\begin{equation}\label{62}
	d^{*}=\min_{\lambda}\max_{\bold{W},\boldsymbol{\xi},\boldsymbol{\rho},\bold{x}}\mathcal{L}(\bold{W},\boldsymbol{\xi},\boldsymbol{\rho},\bold{x}).
	\end{equation}
	For simplicity, we also define:
	\begin{equation}\label{63}
	\mu(\lambda)\triangleq\max_{\bold{W},\boldsymbol{\xi},\boldsymbol{\rho},\bold{x}}\mathcal{L}(\bold{W},\boldsymbol{\xi},\boldsymbol{\rho},\bold{x}).
	\end{equation}
	Based on the weak duality theorem, we have:
	\begin{align}\label{64}
	p^*\leq d^*
	=\min_{\lambda\geq 0}\mu(\lambda).
	\end{align}
	It should be noted that for the feasible set, we have 2 cases as follows:\\
	\textbf{\emph{Case~1}}: One can easily verify that at the optimal point, we have
	\begin{align}\label{65}
	&\mathcal{R}_{1}:\sum_{ i,k\in \mathcal{K}}\sum_{ n \in \mathcal{N}}\bigg( \xi_{i,n}^{k}-(\xi_{i,n}^{k})^{2}\bigg)=0,~
	\\&\mathcal{R}_{2}:\sum_{f\in \mathcal{F}}\sum_{ k\in \mathcal{K}}\sum_{ n \in \mathcal{N}} \bigg( \rho_{k,n}^{f}-(\rho_{k,n}^{f})^{2}\bigg)=0,\\&~\mathcal{R}_{3}:~\sum_{f\in \mathcal{F}}\sum_{ i,k\in \mathcal{K}}\sum_{ n \in \mathcal{N}}  \bigg(x_{i,n}^{k,f}-(x_{i,n}^{k,f})^{2}\bigg)=0. 
	\end{align}
	As a result, $d^*$ is also a feasible solution of (\ref{sic_order_change}). Subsequently, substituting the optimal value of $\lambda$, i.e., $\lambda^{*}$, into the optimization problem (\ref{sic_order_change})~yields
	\begin{equation}\label{66}
	d^*=\mu(\lambda^*)=\max_{\bold{W},\boldsymbol{\xi},\boldsymbol{\rho},\bold{x}}\mathcal{L}(\bold{W},\boldsymbol{\xi},\boldsymbol{\rho},\bold{x})=p^*.
	\end{equation}
	Moreover,~referring to Lagrangian function, in the region $\boldsymbol{\xi},\boldsymbol{\rho},~\bold{x},~\mathcal{R}_{1},~\mathcal{R}_{2},~\mathcal{R}_{3}$, function $\mu(\lambda)$ is a monotonically decreasing function with respect to $\lambda$. On the other hand, it is argued that $d^*=\min_{\lambda \geq 0} \mu(\lambda)$, so we have
	\begin{equation}\label{68}
	\mu(\lambda)= d^*, \forall ~\lambda\geq \lambda^{*}.
	\end{equation}
	This means that for any value of $\lambda\geq\lambda^*$, the solution of~Lagrangian function yields the optimal solution of (\ref{sic_order_change}). \\
	\textbf{\emph{Case~2}}: The second case occurs when~some of integer variables take some values between 0 and 1, causing
	\begin{equation}\label{69}
	\mathcal{R}_{1}>0,~\mathcal{R}_{2}>0,~\mathcal{R}_{3}>0.
	\end{equation}
	Referring to the Lagrangian function and (\ref{63}), at the optimal point, $\mu(\lambda^*)$ tends to $-\infty$. However, this can not happen as it contradicts with primal solution stating that $\mu(\lambda^*)$ is limited from below by the solution of (\ref{sic_order_change}) which is always greater than zero.
	Thus, at the optimal point, we have, $\mathcal{R}_{1}=0,~\mathcal{R}_{2}=0$,~and $\mathcal{R}_{3}=0$.
	}
%\subsection{Proof of Proposition 4}
%Since problem (54) is convex solving the dual problem is equivalent the solving the primal problem which achieves the same solution.~So, Proposition 4 can be analyzed by solving the dual problem (54).~To do so, we write the Lagrangian function with respect to $\boldsymbol{W}$ as follows:
%\begin{align}
%\mathcal{L}(\boldsymbol{W},\alpha,)
%\end{align} 

\bibliographystyle{ieeetr}
%\bibliography{citation_sic}{}
%\vspace{-2em}

\end{document}